\documentclass[a4paper,12pt,fleqn]{article}
\usepackage{amsmath,amssymb}
\usepackage{amsthm}
\usepackage{graphicx}
\usepackage{float}
\usepackage[np]{numprint}
\usepackage{calc}    
\usepackage{tikz,tkz-tab}
\usepackage{braket}
\usepackage{varioref}
\usepackage{hyperref}
\usepackage{lipsum}
\usepackage[left=3cm,right=3cm,top=2.5cm,bottom=4cm]{geometry}
\usepackage{multirow}

\newcommand\swap{\mathtt{SWAP}}
\newcommand\cnot{\mathtt{CNOT}}
\newcommand\PauliX{Pauli-$\mathtt X$ }
\newcommand\PauliY{Pauli-$\mathtt Y$ }
\newcommand\PauliZ{Pauli-$\mathtt Z$ }
\newcommand\cx{\mathtt{CX}}

\newcommand\cz{\mathtt{CZ}}
\newcommand\XX{\mathtt{X}} 
\newcommand\YY{\mathtt{Y}} 
\newcommand\ZZ{\mathtt{Z}} 
\newcommand\II{\mathtt{I}} 
\newcommand\III{\mathrm{I}} 
\newcommand\HH{\mathtt{H}} 
\newcommand\hh{\mathrm{h}} 
\newcommand\PP{\mathtt{P}} 

\renewcommand\phi{\varphi}
\renewcommand\epsilon{\varepsilon}

\renewcommand\leq{\leqslant}

\newcommand\Z{\mathbb{Z}}
\newcommand\C{\mathbb{C}}
\newcommand\F{\mathbb{F}_{2}}
\newcommand\ee{\mathrm e}
\newcommand\ii{\mathrm i}
\newcommand\HS{\mathcal{H}} 
\newcommand\T{\mathrm T} 

\newcommand\symg[1][n]{\mathfrak{S}_{#1}}  
\newcommand\PG[1][n]{\mathcal{E}_{#1}} 
\newcommand\cnotg[1][n]{{\left\langle\cnot\right\rangle}_{#1}} 
\newcommand\czg[1][n]{{\left\langle\cz\right\rangle}_{#1}} 
\newcommand\BG[1][n]{\mathcal{B}_{#1}} 
\newcommand\czxpg[1][n]{\left\langle\PP,\cz,\cnot\right\rangle_{#1}} 
\newcommand\czpg[1][n]{\left\langle\PP,\cz\right\rangle_{#1}} 
\newcommand\SL[1][n]{\mathrm{SL}_{#1}(\mathbb{F}_2)} 
\newcommand\GL[1][n]{\mathrm{GL}_{#1}(\mathbb{F}_2)} 
\newcommand\UG[1][2n]{\mathcal{U}_{2^n}} 

\newcommand\ctozs{\mathtt{C\text{-}to\text{-}ZS}}
\newcommand\ctopzx{\mathtt{C\text{-}to\text{-}PZX}}
\newcommand\ctonf{\mathtt{C\text{-}to\text{-}NF}}
\newcommand\ctorednf{\mathtt{C\text{-}to\text{-}CZredNF}}

\newtheorem{example}{Example}
\newtheorem{theo}[example]{Theorem}
\newtheorem{prop}[example]{Proposition}

\newtheorem{lem}[example]{Lemma}

\newtheorem{defi}[example]{Definition}
\newtheorem{rem}[example]{Remark}
\newtheorem{conj}[example]{Conjecture}

\begin{document}
\setlength\parindent{0mm}

\overfullrule=0mm
\floatstyle{boxed} 
\restylefloat{figure}

\title{Reducing stabilizer circuits without the symplectic group}

\author{Marc Bataille \\ marc.bataille1@univ-rouen.fr \\
  \\ LITIS laboratory, Universit\'e Rouen-Normandie \thanks{685 Avenue de l'Universit\'e, 76800 Saint-\'Etienne-du-Rouvray. France.}}

\date{}

\maketitle

\begin{abstract}
  We start by studying  the subgroup structures  underlying stabilizer circuits. Then we apply our results to provide two normal forms for stabilizer circuits. These forms are computed by induction using simple conjugation rules in the Clifford group and our algorithms do not rely on a special decomposition in the symplectic group. The first normal form has shape  CX-CZ-P-Z-X-H-CZ-P-H, where CX (resp. CZ) denotes a layer of $\cnot$ (resp. $\cz$) gates, P a layer of phase gates, X (resp. Z) a layer of
  \PauliX (resp. \PauliZ) gates. Then we replace most of the $\cz$ gates by $\cnot$ gates to obtain a second normal form of type
  P-CX-CZ-CX-Z-X-H-CZ-CX-P-H. In this second form, both $\cz$ layers have depth 1 and together contain therefore at most $n\ \cz$ gates.
  We also consider normal forms for stabilizer states and graph states. Finally we carry out a few tests on classical and quantum computers in order to show experimentally the utility of these normal forms to reduce the gate count of a stabilizer circuit.
  \end{abstract}

\section{Introduction}
In Quantum Computation, any unitary operation can be approximated to arbitrary accuracy using $\cnot$ gates together with Hadamard, Phase, and  $\pi/8$  gates (see Figure \ref{univers} for a definition of these gates and \cite[Section 4.5.3]{2011NC} for a proof of this result). Therefore, this set of gates is often called the standard set of universal gates. When we restrict this set to Hadamard, Phase and $\cnot$ gates, we obtain the set of Clifford gates. The Pauli group $\PG$ is the group generated by the Pauli gates acting on $n$ qubits (see Figure \ref{Pauli}) and the normalizer of the Pauli Group in the unitary group $\UG$ is called the Clifford group.
In his PhD thesis \cite[Section 5.8]{1997G}, Gottesman gave a constructive proof of the fact that the Clifford gates generate the Clifford group, up to a global phase. He also introduced the stabilizer formalism \cite[Section 10.5.1]{2011NC}, which turned out to be is a very efficient tool to analyze quantum error-correction codes \cite{1997G} and, more generally,  to describe unitary dynamics \cite[Section 10.5.2]{2011NC}. Indeed, the Gottesman-Knill theorem asserts that a stabilizer circuit (\textit{i.e.} a quantum circuit consisting only of Clifford gates) can be simulated efficiently on a classical computer (see \cite[Section 10.5.4]{2011NC} and \cite[p. 52]{1997G}).

Due to their importance in many fields of Quantum Computation, several  normal forms for stabilizer circuits were proposed over the last two decades, with the
aim of  reducing the gate count in these circuits.
The first normal form proposed by Aaronson and Gottesman \cite{2004AG} was successively improved by Maslov and Roetteler \cite{2018MR}, Bravyi and Maslov \cite{2020BM} and
Duncan \textit{et al.} \cite{2020DKPV}. These authors use decomposition methods in the symplectic group over $\F$ in dimension $2n$ \cite{2004AG,2018MR,2020BM} or
ZX-calculus \cite{2020DKPV} in order to compute a normal form. In this paper we provide two normal forms for stabilizer circuits. The first form is similar to the most recent ones \cite{2020BM,2020DKPV} but is computed with a different method that does not rely on the properties of the symplectic matrices. We use and induction process based on conjugation rules in the Clifford group. The second form is obtained from the first one by replacing most of the $\cz$ gates by $\cnot$ gates. This allows to reduce the 2-qubit gate count by using an algorithm proposed in 2004 by Patel \textit{et al.} \cite{2004PMH}.

This article is structured  as follows. Section \ref{background} is a background section on quantum circuits and Clifford gates that will guide the non-specialist reader through the rest of the paper. In Section \ref{groups} we investigate some remarkable subgroups of the Clifford Group. The precise description of these group structures allows us  to propose, in  Section \ref{NF}, a polynomial-time algorithm to compute two types of normal forms. In Section \ref{application} we apply our results to stabilizer states and graph states. Finally,  in section \ref{experience}, we use a C implementation of our algorithms to provide a few statistics which empirically show the interest of these normal forms to reduce the gate count of stabilizer circuits. We also propose an efficient implementation of graph states in the publicly available IBM quantum computers.

\section{Quantum circuits and Clifford gates\label{background}}
In this background section we recall classical notions about quantum circuits and Clifford gates and we also introduce the main notations used in the paper.
In Quantum Information Theory, a qubit is a quantum state that represents the basic information
storage unit. This state is described by a ket vector in the Dirac notation $\ket{\psi} = a_0 \ket{0} + a_1\ket{1}$ where $a_0$ and $a_1$ are complex numbers
such that $|a_0|^2 + |a_1|^2= 1$. The value of $|a_i|^2$ represents the probability that measurement produces
the value $i$. The states $\ket{0}$ and $\ket{1}$ form a basis of the Hilbert space $\HS\simeq \C^2$ where a one qubit quantum system evolves.
Operations on qubits must preserve the norm and are therefore described by unitary operators $U$ in the unitary group $\UG$.
In quantum computation, these operations are represented by quantum gates and a quantum circuit is a conventional representation of the sequence  of quantum gates applied to the qubit register over time. In Figure \ref{Pauli}, we recall the definition of the Pauli gates mentioned in the introduction. Notice that the states $\ket{0}$
and $\ket{1}$ are eigenvectors of the \PauliZ operator respectively associated to the eigenvalues 1 and -1, so the standard computational basis $(\ket{0},\ket{1})$ is also called the $\ZZ$-basis. Notice also that $\XX\ket{0}=1$ and $\XX\ket{1}=0$, hence the \PauliX gate is called the $\mathtt{NOT}$ gate. The phase gate $\PP$ (see Figure \ref{univers}) is defined by $\PP\ket{0}=\ket{0}$ and $\PP\ket{1}=\ii\ket{1}$. The Hadamard gate $\HH$ creates superposition since $\HH\ket{0}=\frac{1}{\sqrt2}(\ket{0}+\ket{1})$. The following useful identities are obtained by direct computation.
\begin{align}
  &\HH^2=\XX^2=\YY^2=\ZZ^2=\II\label{involutions}\\
  &\XX\ZZ=-\ZZ\XX\label{anticom}\\
  &\YY=\ii\XX\ZZ\label{yixz}\\
  &\HH\ZZ\HH=\XX\label{conj-z-h}\\
  &\PP^2=\ZZ\\
  &\PP\XX\PP^{-1}=\YY\label{conj-x-p}
\end{align}
The Pauli group for one qubit is the group generated by the set $\{\XX,\YY,\ZZ\}$. Any element of this group can be written uniquely in the form $\ii^{\lambda}\XX^a\ZZ^b$,
where $\lambda\in\Z_4$ and $a,b\in\F$.
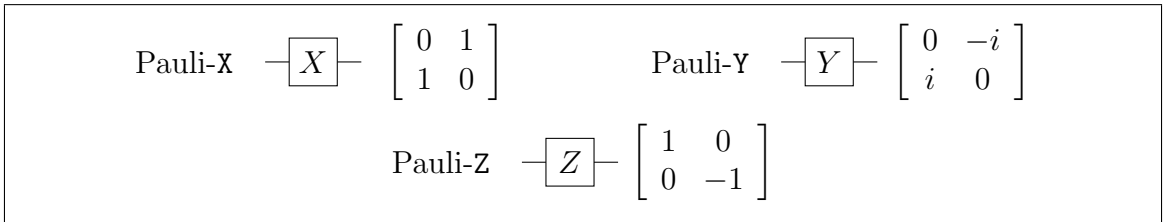
\begin{figure}[h]
	\begin{center}
	\begin{tikzpicture}[scale=1.500000,x=1pt,y=1pt]
\filldraw[color=white] (0.000000, -7.500000) rectangle (24.000000, 7.500000);
\draw[color=black] (0.000000,0.000000) -- (24.000000,0.000000);
\draw[color=black] (0.000000,0.000000) node[left] {\PauliX\ \ };
\draw (45.00, 0.00) node {$\left[\begin{array}{cc}0&1\\1&0\end{array}\right]$};
\begin{scope}
\draw[fill=white] (12.000000, -0.000000) +(-45.000000:8.485281pt and 8.485281pt) -- +(45.000000:8.485281pt and 8.485281pt) -- +(135.000000:8.485281pt and 8.485281pt) -- +(225.000000:8.485281pt and 8.485281pt) -- cycle;
\clip (12.000000, -0.000000) +(-45.000000:8.485281pt and 8.485281pt) -- +(45.000000:8.485281pt and 8.485281pt) -- +(135.000000:8.485281pt and 8.485281pt) -- +(225.000000:8.485281pt and 8.485281pt) -- cycle;
\draw (12.000000, -0.000000) node {$X$};
\end{scope}

\end{tikzpicture}\qquad\qquad
\begin{tikzpicture}[scale=1.500000,x=1pt,y=1pt]
\filldraw[color=white] (0.000000, -7.500000) rectangle (24.000000, 7.500000);
\draw[color=black] (0.000000,0.000000) -- (24.000000,0.000000);
\draw[color=black] (0.000000,0.000000) node[left] {\PauliY\ \ };
\draw (45.00, 0.00) node {$\left[\begin{array}{cc}0&-i\\i&0\end{array}\right]$};
\begin{scope}
\draw[fill=white] (12.000000, -0.000000) +(-45.000000:8.485281pt and 8.485281pt) -- +(45.000000:8.485281pt and 8.485281pt) -- +(135.000000:8.485281pt and 8.485281pt) -- +(225.000000:8.485281pt and 8.485281pt) -- cycle;
\clip (12.000000, -0.000000) +(-45.000000:8.485281pt and 8.485281pt) -- +(45.000000:8.485281pt and 8.485281pt) -- +(135.000000:8.485281pt and 8.485281pt) -- +(225.000000:8.485281pt and 8.485281pt) -- cycle;
\draw (12.000000, -0.000000) node {$Y$};
\end{scope}
\end{tikzpicture}
\begin{tikzpicture}[scale=1.500000,x=1pt,y=1pt]
\filldraw[color=white] (0.000000, -7.500000) rectangle (24.000000, 7.500000);
\draw[color=black] (0.000000,0.000000) -- (24.000000,0.000000);
\draw[color=black] (0.000000,0.000000) node[left] {\PauliZ\ \ };
\draw (45.00, 0.00) node {$\left[\begin{array}{cc}1&0\\0&-1\end{array}\right]$};
\begin{scope}
\draw[fill=white] (12.000000, -0.000000) +(-45.000000:8.485281pt and 8.485281pt) -- +(45.000000:8.485281pt and 8.485281pt) -- +(135.000000:8.485281pt and 8.485281pt) -- +(225.000000:8.485281pt and 8.485281pt) -- cycle;
\clip (12.000000, -0.000000) +(-45.000000:8.485281pt and 8.485281pt) -- +(45.000000:8.485281pt and 8.485281pt) -- +(135.000000:8.485281pt and 8.485281pt) -- +(225.000000:8.485281pt and 8.485281pt) -- cycle;
\draw (12.000000, -0.000000) node {$Z$};
\end{scope}

\end{tikzpicture}\vspace{-4mm}

{ \caption{ The Pauli gates \label{Pauli}}}
\end{center}
\end{figure}

A quantum system of two qubits $A$ and $B$ (also called  a two-qubit register) lives in a 4-dimensional Hilbert space $\HS_A\otimes\HS_B$ and the computational basis of this space is $(\ket{00}=\ket{0}_A\otimes\ket{0}_B,\ket{01}=\ket{0}_A\otimes\ket{1}_B,\ket{10}=\ket{1}_A\otimes\ket{0}_B,\ket{11}=\ket{1}_A\otimes\ket{1}_B)$.
If $U$ is any unitary operator acting on one qubit, a controlled-$U$ gate acts on the Hilbert space $\HS_{A}\otimes\HS_{B}$ as follows.
One of the two qubits (say qubit $A$) is the control qubit whereas the other qubit is the target qubit. If the control qubit $A$ is in the state $\ket 1$ then $U$ is applied on the target qubit $B$ but when qubit $A$ is in the state $\ket{0}$ nothing is done on qubit $B$.
The $\cnot$ gate (or $\cx$ gate) is the controlled-$\XX$ gate with control on qubit $A$ and target on qubit $B$, so the action of $\cnot$ on a two-qubit register is described by :
$\cnot\ket{00}=\ket{00}, \cnot\ket{01}=\ket{01}, \cnot\ket{10}=\ket{11}, \cnot\ket{11}=\ket{10}$ (the corresponding matrix is given in Figure \ref{univers}).
Observe that this action can be sum up by the simple formula   $\cnot\ket{xy}=\ket{x,x \oplus y}$,
where $\oplus$ denotes the XOR operator between two bits $x$ and $y$, which is also the addition in $\F$. In the same way, the reader can check that the controlled-$\ZZ$ operator  acts on a a basis vector as $\cz\ket{xy}=(-1)^{xy}\ket{xy}$.
Note that this action is invariant by switching the control and the target. 
The last 2-qubit gate we need is the $\swap$ gate defined by $\swap\ket{xy}=\ket{yx}$.
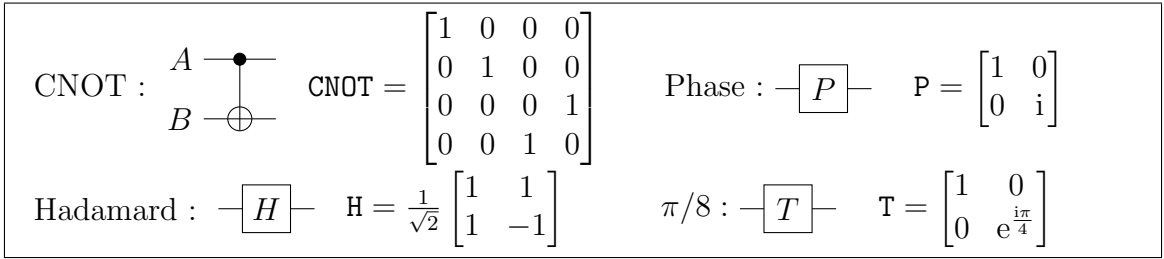
\begin{figure}[h]
\ \ CNOT :\raisebox{-7mm}{
 \begin{tikzpicture}[scale=1.500000,x=1pt,y=1pt]
\filldraw[color=white] (0.000000, -7.500000) rectangle (18.000000, 22.500000);
\draw[color=black] (0.000000,15.000000) -- (18.000000,15.000000);
\draw[color=black] (0.000000,15.000000) node[left] {$A$};
\draw[color=black] (0.000000,0.000000) -- (18.000000,0.000000);
\draw[color=black] (0.000000,0.000000) node[left] {$B$};
\draw (9.000000,15.000000) -- (9.000000,0.000000);
\begin{scope}
\draw[fill=white] (9.000000, 0.000000) circle(3.000000pt);
\clip (9.000000, 0.000000) circle(3.000000pt);
\draw (6.000000, 0.000000) -- (12.000000, 0.000000);
\draw (9.000000, -3.000000) -- (9.000000, 3.000000);
\end{scope}
\filldraw (9.000000, 15.000000) circle(1.500000pt);
\end{tikzpicture}
} 
\ $\cnot=\begin{bmatrix}
  1&0&0&0\\
  0&1&0&0\\
  0&0&0&1\\
  0&0&1&0  
\end{bmatrix}$\qquad
Phase :\raisebox{-3mm}{
\begin{tikzpicture}[scale=1.500000,x=1pt,y=1pt]
\filldraw[color=white] (0.000000, -7.500000) rectangle (24.000000, 7.500000);
\draw[color=black] (0.000000,0.000000) -- (24.000000,0.000000);
\begin{scope}
\draw[fill=white] (12.000000, -0.000000) +(-45.000000:8.485281pt and 8.485281pt) -- +(45.000000:8.485281pt and 8.485281pt) -- +(135.000000:8.485281pt and 8.485281pt) -- +(225.000000:8.485281pt and 8.485281pt) -- cycle;
\clip (12.000000, -0.000000) +(-45.000000:8.485281pt and 8.485281pt) -- +(45.000000:8.485281pt and 8.485281pt) -- +(135.000000:8.485281pt and 8.485281pt) -- +(225.000000:8.485281pt and 8.485281pt) -- cycle;
\draw (12.000000, -0.000000) node {$P$};
\end{scope}
\end{tikzpicture}
}\quad 
$\PP=\begin{bmatrix}1&0\\0&\ii\end{bmatrix}$

\raisebox{-3mm}{
\begin{tikzpicture}[scale=1.500000,x=1pt,y=1pt]
\filldraw[color=white] (0.000000, -7.500000) rectangle (24.000000, 7.500000);
\draw[color=black] (0.000000,0.000000) -- (24.000000,0.000000);
\draw[color=black] (0.000000,0.000000) node[left] {Hadamard\ :\ \ };
\begin{scope}
\draw[fill=white] (12.000000, -0.000000) +(-45.000000:8.485281pt and 8.485281pt) -- +(45.000000:8.485281pt and 8.485281pt) -- +(135.000000:8.485281pt and 8.485281pt) -- +(225.000000:8.485281pt and 8.485281pt) -- cycle;
\clip (12.000000, -0.000000) +(-45.000000:8.485281pt and 8.485281pt) -- +(45.000000:8.485281pt and 8.485281pt) -- +(135.000000:8.485281pt and 8.485281pt) -- +(225.000000:8.485281pt and 8.485281pt) -- cycle;
\draw (12.000000, -0.000000) node {$H$};
\end{scope}
\end{tikzpicture}
}\ \ 
$\HH=\frac{1}{\sqrt2}\begin{bmatrix}1&1\\1&-1\end{bmatrix}$
\hspace{10mm}
$\pi/8 : $\raisebox{-3mm}{
\begin{tikzpicture}[scale=1.500000,x=1pt,y=1pt]
\filldraw[color=white] (0.000000, -7.500000) rectangle (24.000000, 7.500000);
\draw[color=black] (0.000000,0.000000) -- (24.000000,0.000000);
\begin{scope}
\draw[fill=white] (12.000000, -0.000000) +(-45.000000:8.485281pt and 8.485281pt) -- +(45.000000:8.485281pt and 8.485281pt) -- +(135.000000:8.485281pt and 8.485281pt) -- +(225.000000:8.485281pt and 8.485281pt) -- cycle;
\clip (12.000000, -0.000000) +(-45.000000:8.485281pt and 8.485281pt) -- +(45.000000:8.485281pt and 8.485281pt) -- +(135.000000:8.485281pt and 8.485281pt) -- +(225.000000:8.485281pt and 8.485281pt) -- cycle;
\draw (12.000000, -0.000000) node {$T$};
\end{scope}
\end{tikzpicture}
}\quad
$\mathtt{T}=\begin{bmatrix}1&0\\0&\ee^{\frac{\ii\pi}{4}}\end{bmatrix}$

{ \caption{ The standard set of universal gates : names, circuit symbols and matrices. \label{univers}}}
\end{figure}

On a system of $n$ qubits, we label each qubit from 0 to $n-1$ thus following the usual convention. For coherence we also number
the lines and columns of a $n\times n$ matrix from 0 to $n-1$ and we consider that a permutation of the symmetric group $\symg$ is a bijection of $\{0,\cdots,n-1\}$. The $n$-qubit system evolves over time in the Hilbert space $\HS_0\otimes \HS_1\otimes\cdots\otimes \HS_{n-1}$ where $\HS_i$ is the Hilbert space of
qubit $i$.
In this space, a state vector of the standard computational basis is classically denoted by
$\ket{x}=\ket{x_0x_1\cdots x_{n-1}}$ where $x_i$ is in $\{0,1\}$. Let $(e_i)_{0\leq i< n}$ be the canonical basis of $\F^n$. It is convenient to identify the binary label  $x=x_0x_1\cdots x_{n-1}$
with the column vector  $[x_0,\dots,x_{n-1}]^{\T}=\sum_ix_ie_i$ of $\F^n$, which is also denoted by $x$. So the standard basis is $(\ket{x})_{x\in\F^n}$.
When we apply locally a single qubit gate $U$ to qubit $i$ of a $n$-qubit register, the corresponding action on the $n$-qubit system is that of the operator
$U_i=\II\otimes\cdots\otimes\II\otimes U\otimes\II\otimes\cdots\otimes\II=\II^{\otimes i}\otimes U\otimes\II^{\otimes n-i-1}$, 
where $\otimes$ is the Kronecker product of matrices and $\II$ the identity matrix in dimension 2.
As an example, if $n=4$, $\HH_1=\II\otimes\HH\otimes\II\otimes\II$ and $\HH_0\HH_3=\HH\otimes\II\otimes \II\otimes \HH$. We also use a $n$-bit column vector notation,
\textit{e.g.} $\HH_0\HH_3=\HH_{[1,0,0,1]^{\T}}$. Observe that, with this notation, one has $U_i=U_{e_i}$. When $U$ is an involution (\textit{i.e.} $U^2=\II$), the group generated by the $U_i$'s is isomorphic to $\F^n$, since it is an abelian 2-group. This is the case for the Pauli and Hadamard gates but not for the phase gates. For instance
$\HH_{[1,0,0,1]^{\T}}\HH_{[0,0,1,1]^{\T}}=\HH_{[1,0,0,1]^{\T}\oplus[0,0,1,1]^{\T}}=\HH_{[1,0,1,0]^{\T}}=\HH_0\HH_2$. Note that the action of $\ZZ_i$
on $\ket{x}=\ket{x_0\cdots x_{n-1}}$ is described by
$\ZZ_i\ket{x}=(-1)^{x_i}\ket{x}\label{zi}$. Hence, if $v=[v_0,\dots,v_{n-1}]^{\T}\in\F^n$, one has
\begin{equation}
\ZZ_v\ket{x}=(-1)^{v\cdot x}\ket{x},\label{zu}
\end{equation}
where $v\cdot x=\sum_iv_ix_i$. In the same way, $\PP_i\ket{x}=\ii^{x_i}\ket{x}$, hence
\begin{equation}
\PP_v\ket{x}=\ii^{v\cdot x}\ket{x}.\label{pu}
\end{equation} 
A $\cnot$ gate with target on qubit $i$ and control on qubit $j$ will be denoted $X_{ij}$. The reader will pay attention to the fact that our convention is the opposite of that generally used, where $\cnot_{ij}$ denotes a $\cnot$ gate with control on qubit $i$ and target on qubit $j$. The reason for this change will appear clearly in the proof of Theorem \ref{X-GL}. So, if $i<j$, the action of $X_{ij}$ and $X_{ji}$ on a basis vector $\ket{x}$ is given by
\begin{align}
&  X_{ij}\ket{x}=X_{ij}\ket{x_0\cdots x_i\cdots x_j\cdots x_{n-1}}=\ket{x_0\cdots x_i\oplus x_j\cdots x_j\cdots x_{n-1}}\ ,\label{xij}\\
&  X_{ji}\ket{x}=X_{ji}\ket{x_0\cdots x_i\cdots x_j\cdots x_{n-1}}=\ket{x_0\cdots x_i\cdots x_j\oplus x_i\cdots x_{n-1}}.\label{xji}
\end{align}
The $\cz$ (resp. $\swap$) gate between qubits $i$ and $j$ will be denoted by $Z_{ij}$ (resp. $S_{ij}$). Notice that $Z_{ij}=Z_{ji}$ and $S_{ij}=S_{ji}$. These gates are defined by
\begin{align}
  &Z_{ij}\ket{x_0\cdots x_{n-1}}=(-1)^{x_ix_j}\ket{x}\ ,\label{zij}\\
  &S_{ij}\ket{x_0\cdots x_i\cdots x_j\cdots x_{n-1}}=\ket{x_0\cdots x_j\cdots x_i\cdots x_{n-1}}.\label{sij}
\end{align}
Observe that the $X_{ij}$, $S_{ij}$ and $\XX_i$ gates are $2^n\times 2^n$ permutation matrices while the $Z_{ij}$ and $\ZZ_i$  gates are diagonal matrices with all diagonal entries equal to $1$ or $-1$. All these matrices are involutions. The $P_i$ gates are  also diagonal matrices but are not involutions since $\PP_i^2=\ZZ_i$.
We end this section by recalling 3 classical identities. They correspond to the circuit equivalences in Figure \ref{equi}. Each identity can be proved by checking that the action of its left hand side and of its right hand side on a basis vector $\ket{x}$ is the same.
\begin{align}
  X_{ij}&=\HH_i\HH_jX_{ji}\HH_i\HH_j\label{conj-xij-h}\\
  Z_{ij}&=\HH_iX_{ij}\HH_i=\HH_jX_{ji}\HH_j\label{zijxij}\\
  S_{ij}&=X_{ij}X_{ji}X_{ij}=X_{ji}X_{ij}X_{ji}\label{sijxij}
\end{align}
The Pauli group for $n$ qubits is the group generated by the set $\{\XX_i,\YY_i,\ZZ_i\mid i=0\dots n-1\}$. Since
Identities \eqref{involutions}, \eqref{anticom} and \eqref{yixz},  any element of this group can be uniquely written in the form $\ii^{\lambda}\XX_{u}\ZZ_{v}$, where $\lambda\in\Z_4$ and $u,v\in\F^n$. So, using \eqref{anticom}, the multiplication rule in the Pauli group is given by  
\begin{equation}\label{pauli-mult}
  \ii^{\lambda}X_{u}Z_{v}\ii^{\lambda'}X_{u'}Z_{v'}=\ii^{\lambda+\lambda'}(-1)^{u'\cdot v}\XX_{u\oplus u'}\ZZ_{v\oplus v'}.
\end{equation}
A stabilizer circuit for $n$ qubit is an element of the group generated by the set $\{\PP_i,\HH_i,\XX_{ij}\mid 0\leq i,j\leq n-1 \}$. This group contains the $S_{ij}$ and $Z_{ij}$ gates since \eqref{zijxij} and \eqref{sijxij}. It also contains the Pauli group, since $\ZZ_i=\PP_i^2$, $\XX_i=\HH_i\PP_i^2\HH_i$ and $\YY_i=\PP_i\XX_i\PP_i^{-1}=\PP_i\HH_i\PP_i^2\HH_i\PP_i^3$. 
In a stabilizer circuit, changes of the overall phase by a multiple of $\frac{\pi}{4}$ are possible since
\begin{equation}
(\HH_i\PP_i)^3=(\PP_i\HH_i)^3=\ee^{\ii\frac{\pi}{4}}\III.\label{hp3}
\end{equation}
This last equation can be proved by a direct computation.
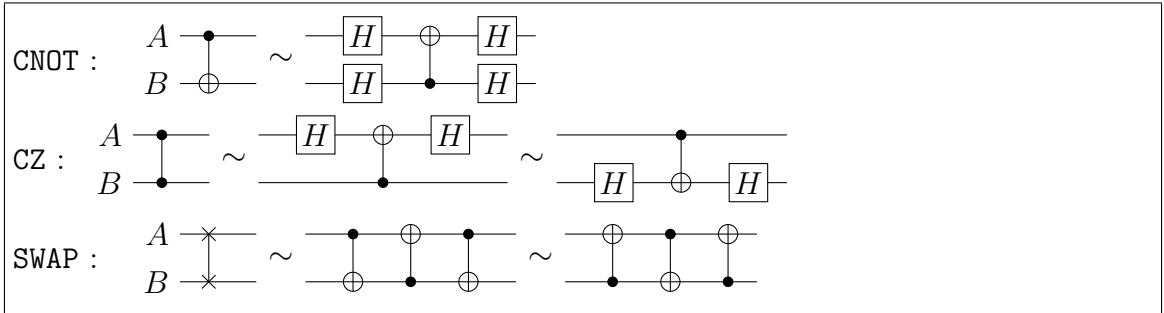
\begin{figure}[h]
$\cnot$ : \ \raisebox{-5mm}{
\begin{tikzpicture}[scale=1.200000,x=1pt,y=1pt]
\filldraw[color=white] (0.000000, -7.500000) rectangle (111.000000, 22.500000);
\draw[color=black] (0.000000,15.000000) -- (111.000000,15.000000);
\draw[color=black] (0.000000,15.000000) node[left] {$A$};
\draw[color=black] (0.000000,0.000000) -- (111.000000,0.000000);
\draw[color=black] (0.000000,0.000000) node[left] {$B$};
\draw (9.000000,15.000000) -- (9.000000,0.000000);
\begin{scope}
\draw[fill=white] (9.000000, 0.000000) circle(3.000000pt);
\clip (9.000000, 0.000000) circle(3.000000pt);
\draw (6.000000, 0.000000) -- (12.000000, 0.000000);
\draw (9.000000, -3.000000) -- (9.000000, 3.000000);
\end{scope}
\filldraw (9.000000, 15.000000) circle(1.500000pt);
\draw[fill=white,color=white] (24.000000, -6.000000) rectangle (39.000000, 21.000000);
\draw (31.500000, 7.500000) node {$\sim$};
\begin{scope}
\draw[fill=white] (57.000000, 15.000000) +(-45.000000:8.485281pt and 8.485281pt) -- +(45.000000:8.485281pt and 8.485281pt) -- +(135.000000:8.485281pt and 8.485281pt) -- +(225.000000:8.485281pt and 8.485281pt) -- cycle;
\clip (57.000000, 15.000000) +(-45.000000:8.485281pt and 8.485281pt) -- +(45.000000:8.485281pt and 8.485281pt) -- +(135.000000:8.485281pt and 8.485281pt) -- +(225.000000:8.485281pt and 8.485281pt) -- cycle;
\draw (57.000000, 15.000000) node {$H$};
\end{scope}
\begin{scope}
\draw[fill=white] (57.000000, -0.000000) +(-45.000000:8.485281pt and 8.485281pt) -- +(45.000000:8.485281pt and 8.485281pt) -- +(135.000000:8.485281pt and 8.485281pt) -- +(225.000000:8.485281pt and 8.485281pt) -- cycle;
\clip (57.000000, -0.000000) +(-45.000000:8.485281pt and 8.485281pt) -- +(45.000000:8.485281pt and 8.485281pt) -- +(135.000000:8.485281pt and 8.485281pt) -- +(225.000000:8.485281pt and 8.485281pt) -- cycle;
\draw (57.000000, -0.000000) node {$H$};
\end{scope}
\draw (78.000000,15.000000) -- (78.000000,0.000000);
\begin{scope}
\draw[fill=white] (78.000000, 15.000000) circle(3.000000pt);
\clip (78.000000, 15.000000) circle(3.000000pt);
\draw (75.000000, 15.000000) -- (81.000000, 15.000000);
\draw (78.000000, 12.000000) -- (78.000000, 18.000000);
\end{scope}
\filldraw (78.000000, 0.000000) circle(1.500000pt);
\begin{scope}
\draw[fill=white] (99.000000, 15.000000) +(-45.000000:8.485281pt and 8.485281pt) -- +(45.000000:8.485281pt and 8.485281pt) -- +(135.000000:8.485281pt and 8.485281pt) -- +(225.000000:8.485281pt and 8.485281pt) -- cycle;
\clip (99.000000, 15.000000) +(-45.000000:8.485281pt and 8.485281pt) -- +(45.000000:8.485281pt and 8.485281pt) -- +(135.000000:8.485281pt and 8.485281pt) -- +(225.000000:8.485281pt and 8.485281pt) -- cycle;
\draw (99.000000, 15.000000) node {$H$};
\end{scope}
\begin{scope}
\draw[fill=white] (99.000000, -0.000000) +(-45.000000:8.485281pt and 8.485281pt) -- +(45.000000:8.485281pt and 8.485281pt) -- +(135.000000:8.485281pt and 8.485281pt) -- +(225.000000:8.485281pt and 8.485281pt) -- cycle;
\clip (99.000000, -0.000000) +(-45.000000:8.485281pt and 8.485281pt) -- +(45.000000:8.485281pt and 8.485281pt) -- +(135.000000:8.485281pt and 8.485281pt) -- +(225.000000:8.485281pt and 8.485281pt) -- cycle;
\draw (99.000000, -0.000000) node {$H$};
\end{scope}
\end{tikzpicture}
}

$\cz$ :\ \raisebox{-5mm}{
\begin{tikzpicture}[scale=1.200000,x=1pt,y=1pt]
\filldraw[color=white] (0.000000, -7.500000) rectangle (204.000000, 22.500000);
\draw[color=black] (0.000000,15.000000) -- (204.000000,15.000000);
\draw[color=black] (0.000000,15.000000) node[left] {$A$};
\draw[color=black] (0.000000,0.000000) -- (204.000000,0.000000);
\draw[color=black] (0.000000,0.000000) node[left] {$B$};
\draw (9.000000,15.000000) -- (9.000000,0.000000);
\filldraw (9.000000, 15.000000) circle(1.500000pt);
\filldraw (9.000000, 0.000000) circle(1.500000pt);
\draw[fill=white,color=white] (24.000000, -6.000000) rectangle (39.000000, 21.000000);
\draw (31.500000, 7.500000) node {$\sim$};
\begin{scope}
\draw[fill=white] (57.000000, 15.000000) +(-45.000000:8.485281pt and 8.485281pt) -- +(45.000000:8.485281pt and 8.485281pt) -- +(135.000000:8.485281pt and 8.485281pt) -- +(225.000000:8.485281pt and 8.485281pt) -- cycle;
\clip (57.000000, 15.000000) +(-45.000000:8.485281pt and 8.485281pt) -- +(45.000000:8.485281pt and 8.485281pt) -- +(135.000000:8.485281pt and 8.485281pt) -- +(225.000000:8.485281pt and 8.485281pt) -- cycle;
\draw (57.000000, 15.000000) node {$H$};
\end{scope}
\draw (78.000000,15.000000) -- (78.000000,0.000000);
\begin{scope}
\draw[fill=white] (78.000000, 15.000000) circle(3.000000pt);
\clip (78.000000, 15.000000) circle(3.000000pt);
\draw (75.000000, 15.000000) -- (81.000000, 15.000000);
\draw (78.000000, 12.000000) -- (78.000000, 18.000000);
\end{scope}
\filldraw (78.000000, 0.000000) circle(1.500000pt);
\begin{scope}
\draw[fill=white] (99.000000, 15.000000) +(-45.000000:8.485281pt and 8.485281pt) -- +(45.000000:8.485281pt and 8.485281pt) -- +(135.000000:8.485281pt and 8.485281pt) -- +(225.000000:8.485281pt and 8.485281pt) -- cycle;
\clip (99.000000, 15.000000) +(-45.000000:8.485281pt and 8.485281pt) -- +(45.000000:8.485281pt and 8.485281pt) -- +(135.000000:8.485281pt and 8.485281pt) -- +(225.000000:8.485281pt and 8.485281pt) -- cycle;
\draw (99.000000, 15.000000) node {$H$};
\end{scope}
\draw[fill=white,color=white] (117.000000, -6.000000) rectangle (132.000000, 21.000000);
\draw (124.500000, 7.500000) node {$\sim$};
\begin{scope}
\draw[fill=white] (150.000000, -0.000000) +(-45.000000:8.485281pt and 8.485281pt) -- +(45.000000:8.485281pt and 8.485281pt) -- +(135.000000:8.485281pt and 8.485281pt) -- +(225.000000:8.485281pt and 8.485281pt) -- cycle;
\clip (150.000000, -0.000000) +(-45.000000:8.485281pt and 8.485281pt) -- +(45.000000:8.485281pt and 8.485281pt) -- +(135.000000:8.485281pt and 8.485281pt) -- +(225.000000:8.485281pt and 8.485281pt) -- cycle;
\draw (150.000000, -0.000000) node {$H$};
\end{scope}
\draw (171.000000,15.000000) -- (171.000000,0.000000);
\begin{scope}
\draw[fill=white] (171.000000, 0.000000) circle(3.000000pt);
\clip (171.000000, 0.000000) circle(3.000000pt);
\draw (168.000000, 0.000000) -- (174.000000, 0.000000);
\draw (171.000000, -3.000000) -- (171.000000, 3.000000);
\end{scope}
\filldraw (171.000000, 15.000000) circle(1.500000pt);
\begin{scope}
\draw[fill=white] (192.000000, -0.000000) +(-45.000000:8.485281pt and 8.485281pt) -- +(45.000000:8.485281pt and 8.485281pt) -- +(135.000000:8.485281pt and 8.485281pt) -- +(225.000000:8.485281pt and 8.485281pt) -- cycle;
\clip (192.000000, -0.000000) +(-45.000000:8.485281pt and 8.485281pt) -- +(45.000000:8.485281pt and 8.485281pt) -- +(135.000000:8.485281pt and 8.485281pt) -- +(225.000000:8.485281pt and 8.485281pt) -- cycle;
\draw (192.000000, -0.000000) node {$H$};
\end{scope}
\end{tikzpicture}
}

$\swap$ : \ \raisebox{-5mm}{
      \begin{tikzpicture}[scale=1.200000,x=1pt,y=1pt]
\filldraw[color=white] (0.000000, -7.500000) rectangle (180.000000, 22.500000);
\draw[color=black] (0.000000,15.000000) -- (180.000000,15.000000);
\draw[color=black] (0.000000,15.000000) node[left] {$A$};
\draw[color=black] (0.000000,0.000000) -- (180.000000,0.000000);
\draw[color=black] (0.000000,0.000000) node[left] {$B$};
\draw (9.000000,15.000000) -- (9.000000,0.000000);
\begin{scope}
\draw (6.878680, 12.878680) -- (11.121320, 17.121320);
\draw (6.878680, 17.121320) -- (11.121320, 12.878680);
\end{scope}
\begin{scope}
\draw (6.878680, -2.121320) -- (11.121320, 2.121320);
\draw (6.878680, 2.121320) -- (11.121320, -2.121320);
\end{scope}
\draw[fill=white,color=white] (24.000000, -6.000000) rectangle (39.000000, 21.000000);
\draw (31.500000, 7.500000) node {$\sim$};
\draw (54.000000,15.000000) -- (54.000000,0.000000);
\begin{scope}
\draw[fill=white] (54.000000, 0.000000) circle(3.000000pt);
\clip (54.000000, 0.000000) circle(3.000000pt);
\draw (51.000000, 0.000000) -- (57.000000, 0.000000);
\draw (54.000000, -3.000000) -- (54.000000, 3.000000);
\end{scope}
\filldraw (54.000000, 15.000000) circle(1.500000pt);
\draw (72.000000,15.000000) -- (72.000000,0.000000);
\begin{scope}
\draw[fill=white] (72.000000, 15.000000) circle(3.000000pt);
\clip (72.000000, 15.000000) circle(3.000000pt);
\draw (69.000000, 15.000000) -- (75.000000, 15.000000);
\draw (72.000000, 12.000000) -- (72.000000, 18.000000);
\end{scope}
\filldraw (72.000000, 0.000000) circle(1.500000pt);
\draw (90.000000,15.000000) -- (90.000000,0.000000);
\begin{scope}
\draw[fill=white] (90.000000, 0.000000) circle(3.000000pt);
\clip (90.000000, 0.000000) circle(3.000000pt);
\draw (87.000000, 0.000000) -- (93.000000, 0.000000);
\draw (90.000000, -3.000000) -- (90.000000, 3.000000);
\end{scope}
\filldraw (90.000000, 15.000000) circle(1.500000pt);
\draw[fill=white,color=white] (105.000000, -6.000000) rectangle (120.000000, 21.000000);
\draw (112.500000, 7.500000) node {$\sim$};
\draw (135.000000,15.000000) -- (135.000000,0.000000);
\begin{scope}
\draw[fill=white] (135.000000, 15.000000) circle(3.000000pt);
\clip (135.000000, 15.000000) circle(3.000000pt);
\draw (132.000000, 15.000000) -- (138.000000, 15.000000);
\draw (135.000000, 12.000000) -- (135.000000, 18.000000);
\end{scope}
\filldraw (135.000000, 0.000000) circle(1.500000pt);
\draw (153.000000,15.000000) -- (153.000000,0.000000);
\begin{scope}
\draw[fill=white] (153.000000, 0.000000) circle(3.000000pt);
\clip (153.000000, 0.000000) circle(3.000000pt);
\draw (150.000000, 0.000000) -- (156.000000, 0.000000);
\draw (153.000000, -3.000000) -- (153.000000, 3.000000);
\end{scope}
\filldraw (153.000000, 15.000000) circle(1.500000pt);
\draw (171.000000,15.000000) -- (171.000000,0.000000);
\begin{scope}
\draw[fill=white] (171.000000, 15.000000) circle(3.000000pt);
\clip (171.000000, 15.000000) circle(3.000000pt);
\draw (168.000000, 15.000000) -- (174.000000, 15.000000);
\draw (171.000000, 12.000000) -- (171.000000, 18.000000);
\end{scope}
\filldraw (171.000000, 0.000000) circle(1.500000pt);
\end{tikzpicture}
}
{ \caption{ Classical equivalences of circuits involving $\cnot$ and Hadamard gates.\label{equi}}}
\end{figure}

\section{Subgroup structures underlying stabilizer circuits\label{groups}}
We start by describing the group $\czg$ which is the group
generated by the $Z_{ij}$ gates. The set of parts of $\{\{i,j\}\mid 0 \leq i<j\leq n-1\}$ is denoted by $\BG$.
As noticed in Section \ref{background}, the  matrices $Z_{ij}$ are involutions and commute with each other since they are diagonal matrices.
So $\czg$ is isomorphic to the abelian 2-group $(\BG,\oplus)$, where $\oplus$ denotes the symmetric difference between two parts of a set (\textit{i.e.} their union minus their intersection). As a consequence, the order of $\czg$ si $2^{\frac{n(n-1)}{2}}$.
For any $B$ in $\BG$, we denote by $Z_B$ the unitary operator in $\czg$ associated to $B$, \textit{i.e.} $Z_B=\prod_{\{i,j\}\in B}Z_{ij}$.
With this notation, Identity \eqref{zij} can be generalized as
\begin{equation}
Z_B\ket{x}=(-1)^{\sum_{\{i,j\}\in B}x_ix_j}\ket{b}.\label{zB}
\end{equation}
To any $B$ in $\BG$ we associate a $n\times n$ $\F$-matrix whose entry $(i,j)$
is 1 if $\{i,j\}$ is in $B$. This matrix is symmetric with only zeros on the diagonal and,  for convenience, we also denote it by $B$. For example $\{\{i,j\}\}$ also denotes the matrix whose entries are all 0
but entries $(i,j)$ and $(j,i)$ that are 1. Depending on the context, it will be clear if we consider the set or the matrix.    Let $q_B$ be the quadratic form defined on $\F^n$ by
\begin{equation}
  q_B(x)=\sum_{\{i,j\}\in B}x_ix_j=\sum_{i<j}B_{ij}x_ix_j,
  \end{equation}
  where $B_{ij}$ is the entry $(i,j)$ of the matrix $B$. Then Identity \eqref{zB} can be rewritten as
\begin{equation}
  Z_B\ket{x}=(-1)^{q_{B}(x)}\ket{x}.
\end{equation}
Note that $B$ can be viewed as the matrix of the alternating (and symmetric) bilinear form associated to the quadratic form $q_B$.\medskip

In a previous work \cite{2020B}, we described  the group $\cnotg$ generated by the $X_{ij}$ gates. We recall now some results of \cite{2020B}. The special linear group on any field $K$ is generated by the set of transvection matrices. In the special case of $K=\F$, this set is reduced to the $n(n-1)$ matrices $T_{ij}=I_n + E_{ij}$, where $E_{ij}$ is the matrix with all entries 0 except the entry $(i,j)$ that is 1.  So $\GL=\SL$ is generated by the matrices $T_{ij}$. The following simple property of the $T_{ij}$ matrices will be of great use.
    \begin{prop}\label{tij-mult}
Multiplying to the left (resp. the right) any matrix $M$ by a 
transvection matrix $T_{ij}$ is equivalent to add the row $j$ (resp. column $i$) to the row $i$ (resp. column $j$) in $M$.
\end{prop}
Applying Proposition \ref{tij-mult} to vector $b\in\F^n$, we can rewrite Relation \eqref{xij} as
\begin{equation}
 X_{ij}\ket{x}=\ket{T_{ij}x}.\label{xijtij}
 \end{equation}
 The above considerations lead quite naturally to the following theorem.

\begin{theo}\label{X-GL}
       
      The group $\cnotg$ generated by the $\cnot$ gates acting on $n$ qubits is isomorphic to $\GL$.
      The morphism $\Phi$ sending each gate $X_{ij}$ to the transvection matrix $T_{ij}$ is an explicit isomorphism.
      The order of $\cnotg$ is $2^{\frac{n(n-1)}{2}}\prod_{i=1}^n(2^i-1)$.
    \end{theo}

    \begin{proof}
      Since the matrices $T_{ij}$ generate $\GL$, it is clear that  $\Phi$ is surjective. For Identity \eqref{xijtij}, a preimage $V$ under $\Phi$ of any matrix $A$ in $\GL$ must satisfy  the relations $V\ket{x}=\ket{Ax}$ for any basis vector $\ket{x}$. Since these relations define a unique matrix $V$, $\Phi$ is injective. The order of $\GL$ is classically obtained by counting the number of basis of the vector space $\F^n$.
\end{proof}

The following conjugation rule between the $\cnot$ gates can be proved by checking, thanks to Identity \eqref{xij}, that the action of its left hand side and of its right hand side on any basis vector $\ket{b}$ is the same.
\begin{equation}
  X_{ij}X_{jk}X_{ij}=X_{jk}X_{ik}\quad (i,j,k \text{ distinct})\label{conj-xij-xjk}
\end{equation}
Let $\czpg$ be the group generated by the set $\{\PP_i,Z_{ij}\}$. 
Any element of the group generated by the $\PP_i$ gates can be written uniquely in the form $\ZZ_v\PP_b$ where $v,b\in\F^n$.
This group is isomorphic to $(\Z_4^n,+)$, one possible isomorphism associating $\ZZ_v\PP_b$ to $2v+b$.
Hence the group $\czpg$ is isomorphic to $\Z_4^n\times\BG$. Any element in $\czpg$ can be written uniquely in the form $\ZZ_v\PP_bZ_B$ and 
\begin{equation}
  \ZZ_v\PP_bZ_B\ZZ_{v'}\PP_{b'}Z_{B'}=\ZZ_{v\oplus v'\oplus bb'}\PP_{b\oplus b'}Z_{B\oplus B'},\label{czpg-mult}
\end{equation}
where $bb'$ denotes the $\F^n$-vector $\sum_ib_ib_i'e_i$.
The conjugation by the $X_{ij}$ gates in $\czpg$ obey to the 7 rules below.
Each equality can be proved by checking, thanks to Identities \eqref{zu} to \eqref{zij}, that the action of its left hand side and of its right hand side on any basis vector $\ket{x}$ is the same.
\begin{align}
  &X_{ij}Z_{ij}X_{ij} = Z_{ij}\ZZ_j\label{conj-Zij}\\
  &X_{ij}Z_{ik}X_{ij} = Z_{ik}Z_{jk}\quad (i,j,k\ \text{distinct})\label{conj-Zik}\\
  &X_{ij}Z_{pq}X_{ij} = Z_{pq}\quad (p,q\neq i)\label{conj-Zpq}\\
  &X_{ij} \ZZ_{i}X_{ij}= \ZZ_i\ZZ_j\label{conj-Zi}\\
  &X_{ij} \ZZ_{j}X_{ij}= \ZZ_j\label{conj-Zj}\\
  &X_{ij} \PP_{i}X_{ij}= \PP_i\PP_jZ_{ij}\label{conj-Pi}\\
  &X_{ij} \PP_{j}X_{ij}= \PP_j\label{conj-Pj}
\end{align}

For any $A$ in $\GL$, let $X_A=\Phi^{-1}(A)$, where $\Phi$ is the morphism defined in Theorem \ref{X-GL}. So, if $[ij]$ denotes the transvection matrix $T_{ij}$, one has $X_{ij}=X_{[ij]}$. Let us denote by $\czxpg$ the group generated by the set $\{\PP_i,Z_{ij},X_{ij}\}$. As described in the following proposition, we can extend relations \eqref{conj-Zij} to \eqref{conj-Pj} to the unitary matrices $\ZZ_v,\PP_b$ or $Z_B$.

\begin{prop}\label{normal}
  The group $\czpg$ is a normal subgroup of $\czxpg$. The conjugation of any element of $\czpg$ by a $\cnot$ gate is described by the relations
\begin{align}
  & X_{[ij]} \ZZ_{v}X_{[ij]}= \ZZ_{[ji]v}\ ,\label{conj-Za-xij}\\
  &X_{[ij]} \PP_{b}X_{[ij]}= \ZZ_{b_ib_je_j}\PP_{[ji]b}Z_{b_i\{\{i,j\}\}}\ ,\label{conj-Pb-xij}\\
  &X_{[ij]}Z_{B}X_{[ij]} = \ZZ_{B_{ij}e_j}Z_{[ji]B[ij]}\ ,\label{conj-ZB-xij}\\
  &X_AZ_BX_A^{-1}=\ZZ_{q_{B}(A^{-1})}Z_{A^{-\T}BA^{-1}}\ ,\label{conj-ZB-XA}
 \end{align}
where $A^{-\T}$ is a shorthand for $\left(A^{\T}\right)^{-1}$, $q_B$ is the quadratic form defined by $B$ and $q_B(A)$ is a shorthand for the vector $[q_B(C_0),\cdots,q_B(C_{n-1})]^{\T}$ with $C_i$ being the column $i$ of $A$. 
\end{prop}

\begin{proof}
  Identities \eqref{conj-Za-xij} and \eqref{conj-Pb-xij} are direct consequences of the conjugation relations \eqref{conj-Zi}, \eqref{conj-Zj}, \eqref{conj-Pi}, \eqref{conj-Pj} and Proposition \eqref{tij-mult} applied to vectors $v$ and $b$.
  Let us prove Identity \eqref{conj-ZB-xij}.
  Let $B_i=\{\{p,q\}\in B\mid i\in\{p,q\}\}$, $B_i^{C}=B\oplus B_i$ and $B_i'=B_i\oplus B_{ij}\{\{i,j\}\}$, then $B=B_{ij}\{\{i,j\}\}\oplus B_i'\oplus B_i^C$.
  On one hand,   $[ji]B[ij]=B_{ij}[ji]\{\{i, j\}\}[ij]\oplus[ji]B_i'[ij]\oplus [ji]B_i^C[ij]$.
  We check that $[ji]\{\{i, j\}\}[ij] = \{\{i, j\}\}$,
  $[ji]\{\{i, k\}\}[ij] = \{\{i, k\}, \{j, k\}\}$ when $k\neq j$ and $[ji]\{\{p, q\}\}[ij] = \{\{p, q\}\}$ when $p,q\neq i$.
  Hence
  \begin{equation}\label{one-hand}
    Z_{[ji]B[ij]}=Z_{ij}^{B_{ij}}Z_{B_i^C}\prod_{k\in K_i}Z_{ik}Z_{jk},
    \end{equation}
  where $K_i=\{k \mid \{i,k\}\in B_i'\}$.
On the other hand, $X_{[ij]}Z_{B}X_{[ij]}=X_{[ij]}Z_{ij}^{B_{ij}}Z_{B_i'}Z_{B_i^C}X_{[ij]}$.
  So, using \eqref{conj-Zij}, \eqref{conj-Zik}
  and \eqref{conj-Zpq}, one has
  \begin{equation}\label{other-hand}
    X_{[ij]}Z_{B}X_{[ij]}=Z_{ij}^{B_{ij}}\ZZ_j^{B_{ij}}Z_{B_i^C}\prod_{k\in K_i}Z_{ik}Z_{jk}.
  \end{equation}
  As $\ZZ_j^{B_{ij}}=\ZZ_{B_{ij}e_j}$, we conclude by comparing \eqref{one-hand} and \eqref{other-hand}.
    Let us prove Identity \eqref{conj-ZB-XA}. Since \eqref{conj-ZB-xij} and \eqref{conj-Za-xij}, it is clear that $X_AZ_BX_A^{-1}$ can be written in the form
  $\ZZ_{v}Z_{A^{-\T}BA^{-1}}$ for some $v$ in $\F^n$, so we have to prove that $v=q_B(A^{-1})$. We start from $\ZZ_{v}=X_AZ_BX_A^{-1}Z_{B'}$, where
  $B'=A^{-\T}BA^{-1}$. Let $\ket{\psi}=\ZZ_{v}\ket{e_i}$, then $\ket{\psi}=(-1)^{v_i}\ket{e_i}$.
  On the other hand, $\ket{\psi}=X_AZ_BX_A^{-1}\ket{e_i}$ since $q_{B'}(e_i)=0$ for any $B'\in\BG$. Besides, $X_A^{-1}\ket{e_i}= \ket{A^{-1}e_i}=\ket{C_i}$ where $C_i$ is the  column $i$ of $A^{-1}$, hence
  $\ket{\psi}=X_AZ_B\ket{C_i}=(-1)^{q_B(C_i)}X_A\ket{C_i}=(-1)^{q_B(C_i)}\ket{e_i}$. Finally we see that $v_i= q_B(C_i)$, thus $v=q_B(A^{-1})$.
  \end{proof}

From Identity \eqref{sijxij}, the gate $S_{ij}$ is in $\cnotg$ and is therefore a $X_A$ gate. Let $(ij)$ be the permutation matrix of $\GL$ associated to the transposition $\tau$ of $\symg$ that swaps   $i$ and $j$, then $(ij)=[ij][ji][ij]=[ji][ij][ji]$ and  $S_{ij}=X_{(ij)}$. The group generated by the $S_{ij}$ gates is a subgroup of $\cnotg$ that is isomorphic to $\symg$. The conjugation by $S_{ij}$ is given by $S_{ij}Z_{pq}S_{ij}=Z_{\tau(p)\tau(q)}$ (see \cite{2019BL} for further development on $\cz$ and $\swap$ gates). As in Proposition \ref{normal}, we prove that
\begin{align}
  & X_{(ij)} \ZZ_{v}X_{(ij)}= \ZZ_{(ij)v}\ , \label{conj-Za-sij}\\
  &X_{(ij)} \PP_{b}X_{(ij)}= \PP_{(ij)b}\ ,\label{conj-Pb-sij}\\
  &X_{(ij)}Z_{B}X_{(ij)} = Z_{(ij)B(ij)}.\label{conj-ZB-sij}
\end{align}
    
The main results of this section are summarized by the theorem that follows.
\begin{theo}\label{decompo}
  Any element of $\czxpg$ admits a unique decomposition in the normal form $\ZZ_v\PP_bZ_BX_A$ where $v,b\in\F^n$, $B\in\BG$, $A\in\GL$.
  The group $\czxpg$ is the semidirect product of its normal subgroup $\czpg$ with $\cnotg$, \textit{i.e.} $\czxpg=\czpg\rtimes\cnotg$.
  The order of $\czxpg$ is therefore $2^{n(n+1)}\prod_{i=1}^n(2^i-1)$. 
\end{theo}

\begin{proof}
  The existence of the decomposition can be proved by induction using Proposition \ref{normal} and Identity \eqref{czpg-mult}. The algorithm in Figure \ref{C-to-PZX} gives the details of the induction step :
  let $C= \prod_{k=1}^{\ell}M_k$, where $M_k\in\{\PP_i,Z_{ij},X_{[ij]}\}$ be an element of $\czxpg$, then the form $\ZZ_v\PP_bZ_BX_A$ for $C$ is the result of
  Algorithm $\ctopzx$  applied to $C$ and $\III$, \textit{i.e.} $\ctopzx(C,\III)$. Observe that the time complexity of this algorithm is only  $O(n\ell)$ since we use row and column additions instead of matrix multiplication thanks to Proposition \ref{tij-mult}.
  Let us prove the unicity of such a decomposition.  Suppose that $\ZZ_v\PP_bZ_BX_A=\ZZ_{v'}\PP_{b'}Z_{B'}X_{A'}$. If $A\neq A'$, there exists $u\in\F^n$ such that $Au\neq A'u$. But this leads to a contradiction
  because $\ZZ_v\PP_bZ_BX_A\ket{u}=\ZZ_{v'}\PP_{b'}Z_{B'}X_{A'}\ket{u}$, so $\ZZ_v\PP_bZ_B\ket{Au}=\ZZ_{v'}\PP_{b'}Z_{B'}\ket{A'u}$, hence $\ket{Au}$ and $\ket{A'u}$ are two different basis vector that are collinear. So $A=A'$ and $\ZZ_v\PP_bZ_B=\ZZ_{v'}\PP_{b'}Z_{B'}$. If there exists $i$ such that $b_i=1$ and $b'_i=0$, then $\PP_b\ket{e_i}=\ii\ket{e_i}$   and $\PP_{b'}\ket{e_i}=\ket{e_i}$, so $\ii\ZZ_vZ_B\ket{e_i}=\ZZ_{v'}Z_{B'}\ket{e_i}$, hence $\ii\ket{e_i}=\pm\ket{e_i}$ which is not possible. Thus $b=b'$. Finally, if $\ZZ_vZ_B=\ZZ_{v'}Z_{B'}$, we show that $v=v'$ and $B=B'$ by comparing their action on $\ket{e_i}$ and $\ket{e_i\oplus e_j}$ for any $i,j$. Since $\czpg$ is a normal subgroup of $\czxpg$, the semidirect product structure is a consequence of the existence and uniqueness of the decomposition. The order of $\czxpg$ is computed using Theorem \ref{X-GL}.  
\end{proof}

\begin{figure}[h]
  $\mathtt{ALGORITHM}$ : Normal form for a stabilizer circuit of $\PP$, $\cz$ and $\cnot$ gates.
  
    $\mathtt{INPUT} :  (C,C')$, where
    
    $\qquad C$ is a circuit of $\ell$ gates in the form $\prod_{k=1}^{\ell}M_k$, where $M_k\in\{\PP_i,Z_{ij},X_{[ij]}\}$,

    $\qquad C'$ is a circuit in the normal form $\ZZ_v\PP_bZ_BX_A$.

    $\mathtt{OUTPUT} :$ $C''$, an equivalent circuit to the product $CC'$ in normal form.

    $\mathtt{1}\quad C''\leftarrow C';$ 

    $\mathtt{2}\quad \mathtt{for}\ k=\ell\ \mathtt{to}\ 1\ \mathtt{do}$
  
    $\mathtt{3}\quad\quad \mathtt{if}\  M_k=Z_{ij}\ \mathtt{then}$
    
    $\mathtt{4}\quad\quad\quad B\leftarrow B \oplus\{\{i,j\}\}\ ; $
    
    $\mathtt{5}\quad\quad \mathtt{else\ if\ } M_k=\PP_i\ \mathtt{then}$
    
    $\mathtt{6}\quad\quad\quad v\leftarrow v\oplus b_ie_i\ ; \ b\leftarrow b \oplus e_i\ ; $
    
    $\mathtt{7}\quad\quad \mathtt{else\ }\ $\quad  // $M_k=X_{[ij]}$ 
    
    $\mathtt{8}\quad\quad\quad v\leftarrow [ji]v \oplus b_ib_je_j\oplus B_{ij}e_j\ ;\ B\leftarrow [ji]B[ij]\oplus b_i\{\{i,j\}\}\ ;b\leftarrow [ji]b\ ;\ A\leftarrow [ij]A\ ;$
    
    $\mathtt{9}\quad \mathtt{return}\ C'';$
    
    \caption{ Algorithm  $\ctopzx$ \label{C-to-PZX}}
  \end{figure}
  
  The  $\ctopzx$ algorithm computes a normal form for particular stabilizer circuits. In \cite{2019BL} we defined the $\ctozs$ algorithm that rewrites any circuit of
  $\swap$ and $\cz$ gates in the form $Z_BS_{\sigma}$, where $S_{\sigma}$ is a circuit of $\swap$ gates corresponding to the permutation $\sigma$ in $\symg$. So the $\ctopzx$ algorithm can be viewed as  an extension of the algorithm $\ctozs$ to the set of input gates $\{\PP_i,Z_{ij},X_{[ij]}\}$. The $\ctopzx$ algorithm is of great use because it is a subroutine called by the main algorithm that computes a normal form for any stabilizer circuits (see next section). The conjugation rules below will also be very helpful.
Let $\hh=\prod_{i=0}^{n-1}\HH_i$, let $U^{\hh}=\hh U\hh^{-1}=\hh U\hh$ for any $U\in\UG$. We use $U^{-\hh}$ as a shorthand notation for $(U^{\hh})^{-1}$.
\begin{align}
  &\hh X_A\hh^{-1}=X_{A^{-\T}}\label{conj-XA-h}\\
  &\PP_i^{\hh}\PP_i\PP_i^{-\hh}=\ee^{\ii\frac{\pi}{4}}\HH_i\XX_i\label{conj-p-ph}\\
  &\PP_i^{\hh}Z_{ik}\PP_i^{-\hh}=Z_{ik}X_{[ik]}\PP_k\label{conj-z-ph}\\
  &Z_{ij}^{\hh}Z_{ij}Z_{ij}^{\hh}=Z_{ij}Z_{ij}^{\hh}Z_{ij}=\HH_i\HH_jX_{(ij)}=X_{(ij)}\HH_i\HH_j\label{conj-zij-zijh}\\
  &Z_{ij}^{\hh}Z_{ik}Z_{ij}^{\hh}=X_{[jk]}Z_{ik}=Z_{ik}X_{[jk]}\quad (i,j,k\text{ distinct})\label{conj-zik-zijh}\\
  &Z_{ij}^{\hh}\PP_jZ_{ij}^{\hh}=\PP_i^{\hh}X_{[ij]}\PP_j\label{conj-pj-zijh}
\end{align}

Identity \eqref{conj-XA-h} is a straightforward consequence of \eqref{conj-xij-h}. Identity \eqref{conj-p-ph} comes from \eqref{hp3}, \eqref{conj-z-h}, \eqref{conj-x-p} and \eqref{yixz}. We prove Identity \eqref{conj-z-ph} using \eqref{conj-Pi}, \eqref{hp3} and \eqref{zijxij}. We prove  Identity \eqref{conj-zij-zijh} using
\eqref{zijxij} and \eqref{sijxij}. Identity \eqref{conj-zik-zijh} is proved using \eqref{zijxij} and \eqref{conj-xij-xjk}.
Identity \eqref{conj-pj-zijh} is proved using \eqref{zijxij} and \eqref{conj-Pi}.

The last identities we need to write a stabilizer circuit in normal form are the conjugation rules of a Pauli product $\XX_{u}\ZZ_{v}$ by the gates $\PP_i$, $X_{[ij]}$, $Z_{ij}$ and $\hh$.
\begin{align}
  &\PP_i\XX_{u}\ZZ_{v}\PP_i^{-1}=\ii^{u_i}\XX_{u}\ZZ_{v\oplus u_i e_i}\label{conj-xz-pi}\\
  &X_{[ij]}\XX_{u}\ZZ_{v}X_{[ij]}=\XX_{[ij]u}\ZZ_{[ji]v}\label{conj-xz-xij}\\
  &Z_{ij}\XX_{u}\ZZ_{v}Z_{ij}=(-1)^{u_iu_j}\XX_{u}\ZZ_{v\oplus u_j e_i\oplus u_i e_j}=(-1)^{u_iu_j}\XX_{u}\ZZ_{v\oplus \{\{i,j\}\}u}\label{conj-xz-zij}\\
  &\hh\XX_u\ZZ_v\hh=\XX_v\ZZ_u\label{conj-xz-h}
\end{align}
Identity \eqref{conj-xz-pi} comes from \eqref{conj-x-p} and \eqref{yixz},
Identity \eqref{conj-xz-xij} from \eqref{conj-Za-xij}, \eqref{conj-z-h} and \eqref{conj-xij-h}, Identity \eqref{conj-xz-zij} from \eqref{conj-xz-xij}, \eqref{anticom} and \eqref{conj-z-h}, Identity \eqref{conj-xz-h} from \eqref{conj-z-h}. 
Finally we generalize below Identities \eqref{conj-xz-pi}, \eqref{conj-xz-xij} and \eqref{conj-xz-zij}.
\begin{align}
  &\PP_b\XX_{u}\ZZ_{v}\PP_b^{-1}=\ii^{\sum_i u_i}\XX_{u}\ZZ_{v\oplus\sum_i b_i u_i e_i}=\ii^{\sum_i u_i}\XX_{u}\ZZ_{v\oplus bu}\label{conj-xz-pb}\\
  &X_{A}\XX_{u}\ZZ_{v}X_{A}^{-1}=\XX_{Au}\ZZ_{A^{-\T}v}\label{conj-xz-XA}\\
  &Z_{B}\XX_{u}\ZZ_{v}Z_{B}=(-1)^{q_B(u)}\XX_{u}\ZZ_{v\oplus Bu}\label{conj-xz-ZB}
\end{align}

\section{Computing normal forms\label{NF}}
In this section we provide algorithms to compute two kinds of normal forms for stabilizer circuits. 
 The first normal form is a generalization of the form $\ZZ_{v}\PP_bZ_BX_A$.

\begin{theo}\label{nf-th} \textbf{Normal form for stabilizer circuits}
  
      Any  stabilizer circuit  $C=\prod_{k=1}^{\ell} M_k $, where $M_k\in\{\PP_i,\HH_i,X_{[ij]}\}$ can be written in polynomial time $O(\ell n^2)$ in the form
     \begin{equation}
     \HH_{\alpha}\PP_dZ_D\HH_{\omega}\ee^{\ii\phi}\XX_{u}\ZZ_{v}\PP_bZ_BX_A\ ,\label{nf}
 \end{equation}

 where     $\alpha,d,\omega,u,v,b\in \F^n$, $D,B\in\BG$, $A\in\GL$, and $\phi\in\{k\frac{\pi}{4},k\in\Z\}$.
    \end{theo}

    \begin{proof} The proof consists in the description of an algorithm that builds the normal form by induction. This algorithm is called the $\ctonf$ algorithm and its skeleton is given in  Figure \ref{C-to-NF}. More precisely, we prove by induction on the length $\ell$ of the input circuit that any  stabilizer circuit  $C=\prod_{k=1}^{\ell} M_k $ can be written in the  normal form $\HH_a\PP_dZ_D\hh\ee^{\ii\phi}\XX_{u}\ZZ_{v}\PP_bZ_BX_A$.
      Then the result comes from a straightforward simplification of $\HH_a$ with $\hh$. We define $\alpha$ and $\omega$ as follows. Let $Q$ be the set of qubits involved in the subcircuit $\PP_dZ_D$, \textit{i.e} $Q=\{i\mid d_i=1\}\cup\{i\mid \exists j, D_{ij}=1\}$. If $a_i=1$ and $i\notin Q$ then $\alpha_i=\omega_i=0$,
      otherwise $\alpha_i=a_i$ and $\omega_i=1$. This simplification corresponds to the subroutine \texttt{simplify} mentioned in Figure \ref{C-to-NF}.
    
    The base case of the induction is clear : if $\ell=0$ then $C=\III$ and a normal form for $C$ is $\HH_a\hh$, where $a=[1,\cdots,1]^{\T}\hh$. To realize the induction step, we prove that, for any $M\in\{\HH_i,\PP_i,X_{[ij]}\}$ and any $C$ in the form $\HH_a\PP_dZ_D\hh\ee^{\ii\phi}\XX_{u}\ZZ_{v}\PP_bZ_BX_A$, the product $MC$ can be written in the same form. The induction step is divided in three cases and a few subcases.
\medskip

The first case corresponds to the subroutine \texttt{merge-Hadamard} mentioned in Figure \ref{C-to-NF}.
\medskip

\texttt{Case 1} : $M=\HH_i$ and $\HH_iC=\HH_{a\oplus e_i}\PP_dZ_D\ee^{\ii\phi}\hh\XX_{u}\ZZ_{v}\PP_bZ_BX_A$, so $\HH_iC$ is in normal form.
\medskip

The second case corresponds to the subroutine \texttt{merge-phase} mentioned in Figure \ref{C-to-NF}.\medskip

\texttt{Case 2} : $M=\PP_i$ and $\PP_iC=\HH_{a}\cdot\HH_a\PP_i\HH_a\cdot\PP_dZ_D\hh\ee^{i\phi}\XX_{u}\ZZ_{v}\PP_b Z_BX_A$ (we use the dot $\cdot$ just to make reading easier).
\medskip

\texttt{Case 2.1} : $a_i=0$. Then $\HH_a\PP_i\HH_a=\PP_i$, so

$\PP_iC\stackrel{\eqref{czpg-mult}}{=}\HH_{a}\ZZ_{d_i e_i}\PP_{d\oplus e_i}Z_D\hh\ee^{\ii\phi}\XX_{u}\ZZ_{v}\PP_bZ_BX_A\stackrel{\eqref{conj-z-h}}{=}
\HH_{a}\PP_{d\oplus e_i}Z_D\hh\ee^{\ii\phi}\XX_{u\oplus d_i e_i}\ZZ_{v}\PP_bZ_BX_A$  and $\PP_iC$ is in normal form.

\medskip

\texttt{Case 2.2} : $a_i=1$. Then $\HH_a\PP_i\HH_a=\PP_i^{\hh}$, so

$\PP_iC=\HH_{a}\cdot\PP_i^{\hh}\PP_{d}\PP_i^{-\hh}\cdot\PP_i ^{\hh}Z_D\PP_i^{-\hh}\cdot\hh\cdot \PP_i\ee^{\ii\phi}\XX_{u}\ZZ_{v}\PP_i^{-1}\cdot \PP_i\PP_bZ_BX_A$,

$\PP_iC\stackrel{\eqref{conj-xz-pi},\eqref{czpg-mult}}{=}\HH_{a}\cdot\PP_i^{\hh}\PP_{d}\PP_i^{-\hh}\cdot\PP_i^{\hh} Z_D\PP_i^{-\hh}\cdot
\hh\cdot\ee^{\ii\phi}\ii^{u_i}\XX_{u}\ZZ_{v\oplus u_i e_i}\cdot \ZZ_{b_i e_i}\PP_{b\oplus e_i}Z_BX_A$.

Let $\phi'=\phi + u_i\frac\pi2$, $u'=u$, $v'=v\oplus u_i e_i\oplus b_i e_i$ and $b'=b\oplus e_i$ then

$\PP_iC=\HH_{a}\cdot\PP_i^{\hh}\PP_{d}\PP_i^{-\hh}\cdot\PP_i^{\hh} Z_D\PP_i^{-\hh}\cdot
\hh\cdot\ee^{\ii\phi'}\XX_{u'}\ZZ_{v'}\cdot\PP_{b'}Z_BX_A$.\medskip

\texttt{Case 2.2.1} : $d_i=0$.  Then $\PP_i^{\hh}\PP_{d}\PP_i^{-\hh}=\PP_{d}$.

Let $D_i=\{\{p,q\}\in D\mid i\in\{p,q\}\}$ and $K_i=\{k\mid \{k,i\}\in D\}$, then

$\PP_i^{\hh} Z_D\PP_i^{-\hh}=Z_{D\oplus D_i}\PP_i^{\hh}\cdot \prod_{k\in K_i}Z_{ik}\cdot\PP_i^{-\hh}
\stackrel{\eqref{conj-z-ph}}{=}Z_{D\oplus D_i}\prod_{k\in K_i}\left(Z_{ik}X_{ik}\PP_k\right)$, hence

$\PP_iC=\HH_{a}\PP_{d}\cdot Z_{D\oplus D_i}\prod_{k\in K_i}\left(Z_{ik}X_{ik}\PP_k\right)\cdot
\hh\ee^{\ii\phi'}\XX_{u'}\ZZ_{v'}\cdot\PP_{b'}Z_BX_A$.

Let $\ZZ_{w'}\PP_{d'}Z_{D'}X_{A'}=\ctopzx(\PP_{d}Z_{D\oplus D_i}\prod_{k\in K_i}Z_{ik}X_{ik}\PP_k,\III)$, then $A'=\prod_{k\in K_i}[ik]$ and
 
$\PP_iC=\HH_{a}\ZZ_{w'}\PP_{d'}Z_{D'}X_{A'}\hh\ee^{\ii\phi'}\XX_{u'}\ZZ_{v'}\PP_{b'}Z_BX_A$. 

So $\PP_iC\stackrel{\eqref{conj-XA-h}}{=}\HH_{a}\ZZ_{w'}\PP_{d'}Z_{D'}\hh\cdot X_{A'^{-\T}}\ee^{\ii\phi'}\XX_{u'}\ZZ_{v'}X_{A'^{-\T}}^{-1}\cdot X_{A'^{-\T}}\PP_{b'}Z_BX_A$, and $A'^{-\T}=\prod_{k\in K_i}[ki]$.

Let $\ee^{\ii\phi''}\XX_{u''}\ZZ_{v''}\stackrel{\eqref{conj-xz-XA}}{=}X_{A'^{-\T}}\ee^{\ii\phi'}\XX_{u'}\ZZ_{v'}X_{A'^{-\T}}^{-1}$.

Let $\ZZ_{w''}\PP_{b''}Z_{B''}X_{A''}=\ctopzx\left(X_{A'^{-\T}},\PP_{b'}Z_BX_A\right)$.

Then $\PP_iC=\HH_{a}\ZZ_{w'}\PP_{d'}Z_{D'}\hh\ee^{\ii\phi''}\XX_{u''}\ZZ_{v''}\ZZ_{w''}\PP_{b''}Z_{B''}X_{A''}$, so

$\PP_iC\stackrel{\eqref{conj-z-h}}{=}\HH_{a}\PP_{d'}Z_{D'}\hh\ee^{\ii\phi''}\XX_{u''\oplus w'}\ZZ_{v''\oplus w''}\PP_{b''}Z_{B''}X_{A''}$, and $\PP_iC$ is in normal form.
\medskip

\texttt{Case 2.2.2} :  $d_i=1$.  Then $\PP_i^{\hh}\PP_{d}\PP_i^{-\hh}=\PP_i^{\hh}\PP_i\PP_i^{-\hh}\cdot\PP_{d\oplus e_i}
\stackrel{\eqref{conj-p-ph}}{=}\ee^{\ii\frac{\pi}{4}}\HH_i\XX_i\cdot\PP_{d\oplus e_i}$, so

$\PP_iC=\HH_{a}\cdot\ee^{\ii\frac{\pi}{4}}\HH_i\XX_i\PP_{d\oplus e_i}\cdot\PP_i^{\hh} Z_D\PP_i^{-\hh}\cdot\hh\cdot\ee^{\ii\phi'}\XX_{u'}\ZZ_{v'}\cdot\PP_{b'}Z_BX_A$,

$\PP_iC=\HH_{a\oplus e_i}\PP_{d\oplus e_i}\XX_i\cdot\PP_i^{\hh} Z_D\PP_i^{-\hh}\cdot\hh\cdot\ee^{\ii(\phi'+\frac{\pi}{4})}\XX_{u'}\ZZ_{v'}\cdot\PP_{b'}Z_BX_A$.

$\PP_iC=\HH_{a\oplus e_i}\PP_{d\oplus e_i}\PP_i^{\hh} Z_D\PP_i^{-\hh}\cdot(\PP_i^{\hh} Z_D\PP_i^{-\hh})^{-1}\XX_i(\PP_i^{\hh} Z_D\PP_i^{-\hh})\cdot
\hh\cdot\ee^{\ii(\phi'+\frac{\pi}{4})}\XX_{u'}\ZZ_{v'}\cdot\PP_{b'}Z_BX_A$.

Using Identity \eqref{conj-xz-pi},\eqref{conj-xz-h},\eqref{conj-xz-ZB} and \eqref{pauli-mult}, we compute $\phi'',u'',v''$ such that

$\PP_iC=\HH_{a\oplus e_i}\PP_{d\oplus e_i}\PP_i^{\hh} Z_D\PP_i^{-\hh}\hh\ee^{\ii\phi''}\XX_{u''}\ZZ_{v''}\PP_{b'}Z_BX_A$.
Then we proceed as in Case 2.2.1. and we obtain a normal form for $\PP_iC$

\medskip
The third and last case corresponds to the subroutine \texttt{merge-CNOT} mentioned in Figure \ref{C-to-NF}.
\medskip

\texttt{Case 3} : $M=X_{[ij]}$.

$X_{[ij]}C=\HH_{a}\cdot\HH_aX_{[ij]}\HH_a\cdot\PP_dZ_D\hh\ee^{i\phi}\XX_{u}\ZZ_{v}\PP_bZ_BX_A$.
\medskip

\texttt{Case 3.1} : $a_i=a_j=0$.

Then $\HH_aX_{[ij]}\HH_a=X_{[ij]}$ and
$X_{[ij]}C=\HH_{a}X_{[ij]}\PP_dZ_D \hh\ee^{i\phi}\XX_{u}\ZZ_{v}\PP_bZ_BX_A$.

Let $\ZZ_{w'}\PP_{d'}Z_{D'}X_{[ij]}=\ctopzx(X_{[ij]},\PP_dZ_D)$ , then

$X_{[ij]}C=\HH_{a}\ZZ_{w'}\PP_{d'}Z_{D'}X_{[ij]}\hh\ee^{i\phi}\XX_{u}\ZZ_{v}\PP_bZ_BX_A$,

$X_{[ij]}C\stackrel{\eqref{conj-xij-h}}{=}\HH_{a}\ZZ_{w'}\PP_{d'}Z_{D'}\cdot\hh\cdot X_{[ji]}\ee^{i\phi}\XX_{u}\ZZ_{v}X_{[ji]}\cdot X_{[ji]}\PP_bZ_BX_A$,

$X_{[ij]}C\stackrel{\eqref{conj-xz-xij}}{=}\HH_{a}\ZZ_{w'}\PP_{d'}Z_{D'}\cdot\hh\cdot\ee^{i\phi}\XX_{[ji]u}\ZZ_{[ij]v}\cdot X_{[ji]}\PP_bZ_BX_A$,

$X_{[ij]}C\stackrel{\eqref{conj-z-h}}{=}\HH_{a}\PP_{d'}Z_{D'}\cdot\hh\cdot\ee^{i\phi}\XX_{[ji]u\oplus w'}\ZZ_{[ij]v}\cdot X_{[ji]}\PP_bZ_BX_A$.

Let $\ZZ_{w''}\PP_{b''}Z_{B''}X_{A''}=\ctopzx(X_{[ji]},\PP_bZ_BX_A)$, then

$X_{[ij]}C=\HH_{a}\PP_{d'}Z_{D'}\hh\ee^{i\phi}\XX_{[ji]u\oplus w'}\ZZ_{[ij]v\oplus w''}\PP_{b''}Z_{B''}X_{A''}$ and $X_{[ij]}C$ is in normal form.

\medskip

\texttt{Case 3.2} : $a_i=a_j=1$.

Then $\HH_aX_{[ij]}\HH_a\stackrel{\eqref{conj-xij-h}}{=}X_{[ji]}$ and we proceed as in case 3.1, swapping $i$ and $j$.
\medskip

\texttt{Case 3.3} : $a_i=1$ and $a_j=0$.

Then $\HH_aX_{[ij]}\HH_a\stackrel{\eqref{zijxij}}{=}Z_{ij}$ and a normal form for $X_{[ij]}C$ is

$\HH_{a}\PP_dZ_{D\oplus\{\{i,j\}\}}\hh\ee^{i\phi}\XX_{u}\ZZ_{v}\PP_bZ_BX_{A}$.
\medskip

\texttt{Case 3.4} : $a_i=0$ and $a_j=1$.

Then $\HH_aX_{[ij]}\HH_a=Z_{ij}^{\hh}$ and
$X_{[ij]}C=\HH_{a}Z_{ij}^{\hh}\PP_dZ_D\hh\ee^{i\phi}\XX_{u}\ZZ_{v}\PP_bZ_BX_{A}$.
\medskip

\texttt{Case 3.4.1} : $D_{ij}=0$.

$X_{[ij]}C=\HH_{a}\cdot Z_{ij}^{\hh}\PP_dZ_{ij}^{\hh}\cdot Z_{ij}^{\hh}Z_DZ_{ij}^{\hh}\cdot \hh \cdot Z_{ij} \ee^{i\phi} \XX_{u}\ZZ_{v}Z_{ij}\cdot Z_{ij}\PP_bZ_{B}X_{A}$

$X_{[ij]}C\stackrel{\eqref{conj-xz-zij}}{=}\HH_{a}\cdot Z_{ij}^{\hh}\PP_dZ_{ij}^{\hh}\cdot Z_{ij}^{\hh}Z_DZ_{ij}^{\hh}\cdot
\hh \cdot \ee^{i\phi}(-1)^{u_iu_j}\XX_{u}\ZZ_{v\oplus u_i e_j \oplus u_j e_i}\cdot\PP_bZ_{B\oplus \{\{i,j\}\}}X_{A}$

Let $\phi'=\phi + u_iu_j\pi$, $u'=u$, $v'=v\oplus u_i e_j \oplus u_j e_i $ and $B'=B\oplus \{\{i,j\}\}$, then

$X_{[ij]}C=\HH_{a}\cdot Z_{ij}^{\hh}\PP_dZ_{ij}^{\hh}\cdot Z_{ij}^{\hh}Z_DZ_{ij}^{\hh}\cdot
\hh \cdot \ee^{i\phi'}\XX_{u'}\ZZ_{v'}\cdot\PP_bZ_{B'}X_{A}$.

Let $D_i=\{\{p,q\}\in D\mid i\in\{p,q\}\}$ and $K_i=\{k\mid \{i,k\}\in D\}$, then $D_i\cap D_j=\emptyset$ since $D_{ij}=0$ and 
$Z_{ij}^{\hh}Z_DZ_{ij}^{\hh}=Z_{D\oplus D_i\oplus D_j}\cdot Z_{ij}^{\hh}\cdot\prod_{k\in K_i}Z_{ik}\cdot\prod_{k\in K_j}Z_{jk}\cdot Z_{ij}^{\hh}$, so

$Z_{ij}^{\hh}Z_DZ_{ij}^{\hh}\stackrel{\eqref{conj-zik-zijh}}{=}Z_{D\oplus D_i\oplus D_j}\cdot\prod_{k\in K_i}Z_{ik}X_{[jk]}\cdot\prod_{k\in K_j}Z_{jk}X_{[ik]}$.

Let $\ZZ_{w'}Z_{D'}X_{A'}=\ctopzx(Z_{ij}^{\hh}Z_DZ_{ij}^{\hh},\III)$, then $A'=\prod_{k\in K_i}[jk]\prod_{k\in K_j}[ik]$ and

$X_{[ij]}C=\HH_{a}\cdot Z_{ij}^{\hh}\PP_dZ_{ij}^{\hh}\cdot \ZZ_{w'}Z_{D'}X_{A'} \cdot\hh\cdot  \ee^{i\phi'}\XX_{u'}\ZZ_{v'}\cdot\PP_bZ_{B'}X_{A}$.

So $X_{[ij]}C\stackrel{\eqref{conj-XA-h}}{=}\HH_{a}\cdot Z_{ij}^{\hh}\PP_dZ_{ij}^{\hh}\cdot\ZZ_{w'}Z_{D'}\cdot \hh\cdot
X_{A'^{-\T}}\ee^{i\phi'} \XX_{u'}\ZZ_{v'}X_{A'^{-\T}}^{-1}\cdot
X_{A'^{-\T}}\PP_bZ_{B'}X_{A}$,

and $A'^{-\T}=\prod_{k\in K_i}[kj]\prod_{k\in K_j}[ki]$.

Let $\ee^{i\phi''}\XX_{u''}\ZZ_{v''}\stackrel{\eqref{conj-xz-XA}}{=} X_{A'^{-\T}}\ee^{i\phi'}\XX_{u'}\ZZ_{v'}X_{A'^{-\T}}^{-1}$.

Let $\ZZ_{w''}\PP_{b''}Z_{B''}X_{A''}=\ctopzx(X_{A'^{-\T}},\PP_bZ_{B'}X_{A})$.

Then $X_{[ij]}C=\HH_{a}\cdot Z_{ij}^{\hh}\PP_dZ_{ij}^{\hh}\cdot\ZZ_{w'}Z_{D'} \hh
\ee^{i\phi''}\XX_{u''}\ZZ_{v''}\ZZ_{w''}\PP_{b''}Z_{B''}X_{A''}$, so

$X_{[ij]}C\stackrel{\eqref{conj-z-h}}{=}\HH_{a}\cdot Z_{ij}^{\hh}\PP_dZ_{ij}^{\hh}\cdot Z_{D'} \hh
\ee^{i\phi''}\XX_{u''\oplus w'}\ZZ_{v''\oplus w''}\PP_{b''}Z_{B''}X_{A''}$.

\medskip

\texttt{Case 3.4.1.1} : $d_i=d_j=0$.

Then $Z_{ij}^{\hh}\PP_dZ_{ij}^{\hh}=\PP_d$ and
$X_{[ij]}C=\HH_{a}\PP_d Z_{D'} \hh \ee^{i\phi''}\XX_{u''\oplus w'}\ZZ_{v''\oplus w''} \PP_{b''}Z_{B''}X_{A''}$,

so $X_{ij}C$ is in normal form.
\medskip

\texttt{Case 3.4.1.2} : $d_i=0$ and $d_j=1$.

Then $Z_{ij}^{\hh}\PP_dZ_{ij}^{\hh}=Z_{ij}^{\hh}\PP_jZ_{ij}^{\hh}\PP_{d\oplus e_j}\stackrel{\eqref{conj-pj-zijh}}{=}\PP_i^{\hh}X_{[ij]}\PP_j\PP_{d\oplus e_j}
=\PP_i^{\hh}X_{[ij]}\PP_{d}$

Hence $X_{[ij]}C=\HH_{a}\cdot\PP_i^{\hh}X_{[ij]}\PP_{d} \cdot Z_{D'} \hh
\ee^{i\phi''}\XX_{u''\oplus w'}\ZZ_{v''\oplus w''}\PP_{b''}Z_{B''}X_{A''}$.

Observe that $\PP_{d} Z_{D'} \hh\ee^{i\phi''}\XX_{u''\oplus w'}\ZZ_{v''\oplus w''}\PP_{b''}Z_{B''}X_{A''}$ is already in normal form, so we merge 
$X_{[ij]}$ with this form using case 3.1. As no Hadamard gate is created in case 3.1, we can use case 2.2 to merge $\PP_i^{\hh}$ . Finally we merge $\HH_a$ using case 1
and obtain thereby a  normal form for $X_{[ij]}C$.

\medskip

\texttt{Case 3.4.1.3} : $d_i=1$ and $d_j=0$.

We proceed as in case 3.4.1.2, swapping $i$ and $j$.

\medskip

\texttt{Case 3.4.1.4} : $d_i=d_j=1$.

Then $Z_{ij}^{\hh}\PP_dZ_{ij}^{\hh}=Z_{ij}^{\hh}\PP_j\PP_iZ_{ij}^{\hh}\PP_{d\oplus e_i\oplus e_j}
\stackrel{\eqref{conj-pj-zijh}}{=}\PP_i^{\hh}X_{[ij]}\PP_j \PP_j^{\hh}X_{[ji]}\PP_i\PP_{d\oplus e_i\oplus e_j}$.

Since Identity \eqref{conj-Pj}, $\PP_j^{\hh}X_{[ji]}$ and $\PP_i\PP_{d\oplus e_i\oplus e_j}$ commutes, so 
$Z_{ij}^{\hh}\PP_dZ_{ij}^{\hh}=\PP_i^{\hh}X_{[ij]}\PP_d\PP_j^{\hh}X_{[ji]}$.

Hence $X_{[ij]}C=\HH_{a}\cdot \PP_i^{\hh}X_{[ij]}\PP_d\PP_j^{\hh}X_{[ji]}\cdot Z_{D'} \hh
\ee^{i\phi''}\XX_{u''\oplus w'}\ZZ_{v''\oplus w''}\PP_{b''}Z_{B''}X_{A''}$.

Observe that $Z_{D'} \hh\ee^{i\phi''}\XX_{u''\oplus w'}\ZZ_{v''\oplus w''}\PP_{b''}Z_{B''}X_{A''}$ is already in normal form, so we  merge $X_{[ji]}$ with this
form using case 3.1. As no Hadamard gate is created in case 3.1, we can use cases 2.2 and 2.2.1 to merge $\PP_j^{\hh}$.  As no Hadamard gate is created in case
2.2.1, we can use case 2.1 to merge $\PP_d$ and case 3.1 to merge  $X_{[ij]}$. Again, no Hadamard gate is created in cases 2.1 and 3.1, so we finally merge  $\PP_i^{\hh}$ (case 2.2) and $\HH_a$ (case 1) and obtain a normal form for $X_{ij}C$.

\medskip

\texttt{Case 3.4.2} : $D_{ij}=1$.

$X_{[ij]}C=\HH_{a}\cdot Z_{ij}^{\hh}Z_{ij}\cdot\PP_dZ_{D'}\hh\ee^{i\phi}\XX_{u}\ZZ_{v}\PP_bZ_BX_{A}$, where $D'=D\oplus \{\{i,j\}\}$, hence $D'_{ij}=0$.

$X_{[ij]}C\stackrel{\ref{conj-zij-zijh}}{=}\HH_{a}\cdot \HH_i\HH_jX_{(ij)}Z_{ij}^{\hh}  \cdot\PP_dZ_{D'}\hh\ee^{i\phi}\XX_{u}\ZZ_{v}\PP_bZ_BX_{A}$.

We use the conjugations rules \eqref{conj-Za-sij}, \eqref{conj-Pb-sij} and \eqref{conj-ZB-sij} by the $\swap$ gate $X_{(ij)}$ and we merge the Hadamard gates to obtain

$X_{[ij]}C=\HH_{a\oplus e_i\oplus e_j}Z_{ij}^{\hh} \PP_{(ij)d}Z_{(ij)D'(ij)}\hh\ee^{i\phi}\XX_{(ij)u}\ZZ_{(ij)v}\PP_{(ij)b}Z_{(ij)B(ij)}X_{(ij)A}$.

Finally we proceed as in case 3.4.1, since $\{i,j\}\not\in (ij)D'(ij)$.
\medskip

Let us compute the worst case time complexity of the $\ctonf$ algorithm. 
The complexity of the $\ctopzx$ algorithm  is $O(\ell n)$, where $\ell$ is the gate count of the input circuit. As we apply the $\ctopzx$ algorithm to circuits of $O(n)$ gates
(cases 2.2 and 3.4.1), the cost of each call to $\ctopzx$ is  $O(n^2)$ operations. 
Conjugating a Pauli gate $\XX_u\ZZ_v$ by $X_A$  requires only $O(n)$ operations because we use a decomposition in less than $n$ (resp. $2n$) transvections of matrix $A$ in case 2.2.1 (resp. case 3.4.1) and each conjugation by a transvection costs $O(1)$ operations \eqref{conj-xz-xij}. Conjugating a Pauli gate $\XX_i$ by $(\PP_i^{\hh} Z_B\PP_i^{-\hh})^{-1}$ (case 2.2.2) requires $O(n^2)$ operations.
Other rewritings require $O(n)$ or $O(1)$ operations. Therefore the cost of merging a $M_k$ gate in the normal form at each induction step is $O(n^2)$ operations.
As we need $\ell$ steps to write $C$ in normal form, we see that the number of operations performed by the $\ctonf$ algorithm is $O(\ell n^2)$.
\end{proof}

  \begin{rem}
    The global phase $\phi$ of a quantum circuit is generally considered as unimportant  since it is physically unobservable. However, we decided  not to neglect it during the computation process of the normal form since knowing its exact value is, at least, of purely mathematical interest. Moreover, calculating the exact value of $\phi$ does not require much additional work. 
  \end{rem}

  \begin{rem}
    The $\ctonf$ algorithm can also take $\cz$, $\swap$, $\ZZ$, $\XX$, or $\YY$ gates as input since $Z_{ij}=\HH_iX_{ij}\HH_i$,
    $S_{ij}=X_{[ij]}X_{[ji]}X_{[ij]}$, $\ZZ_i=\PP_i^2$, $\XX_i=\HH_i\PP_i^2\HH_i$ and $\YY_i=\PP_i\XX_i\PP_i^{-1}=\PP_i\HH_i\PP_i^2\HH_i\PP_i^3$.
    \end{rem}

\begin{rem}
  Notice that the form $\ZZ_v\PP_bZ_BX_A$ defined in Theorem \ref{decompo} is just a particular case of the normal form $\HH_{\alpha}\PP_dZ_D\HH_{\omega}\ee^{\ii\phi}\XX_{u}\ZZ_{v}\PP_bZ_BX_A$. The $\ctonf$ algorithm can be applied to an input circuit of type $\prod_{k=1}^{\ell}M_k$, where $M_k\in\{\PP_i, Z_{ij},X_{[ij]}\}$ and yields in this case a circuit in the normal form  $\ZZ_v\PP_bZ_BX_A$. In this sens, the $\ctonf$ algorithm is an extension of the $\ctopzx$ algorithm to any stabilizer circuit.
\end{rem}

\begin{rem} Unlike the particular normal form $\ZZ_{v}\PP_bZ_BX_A$ computed by the $\ctopzx$ algorithm, the general normal form
  $\HH_{\alpha}\PP_dZ_D\HH_{\omega}\ee^{\ii\phi}\XX_{u}\ZZ_{v}\PP_bZ_BX_A$ is not unique. For instance, since Identity \ref{conj-zij-zijh}, one has $\HH_i\HH_jZ_{ij}\HH_i\HH_jZ_{ij}=Z_{ij}\HH_i\HH_jX_{(ij)}$.
  \end{rem}

  \begin{figure}[h]

    $\mathtt{ALGORITHM\ :}$  Normal form for a stabilizer circuit.
    
    $\mathtt{INPUT\ :}$  C, a circuit of length $\ell$ in the form $\prod_{k=1}^{\ell}M_k$, where $M_k\in\{\HH_i, \PP_i, X_{[ij]}\}$.

    $\mathtt{OUTPUT\ :}$ NF, an equivalent circuit to C in the normal form  $\HH_{\alpha}\PP_dZ_D\HH_{\omega}\ee^{\ii\phi}\XX_{u}\ZZ_{v}\PP_bZ_BX_A$.

    $\mathtt{1}$\quad\ $nv\leftarrow \text{null vector}$; $nm\leftarrow \text{null matrix}$; $I\leftarrow \text{identity matrix}$; 

    $\mathtt{2}$\quad\ $\alpha\leftarrow [1,\cdots,1]^{\T}$; $\omega\leftarrow [1,\cdots,1]^{\T}$;

    $\mathtt{3}$\quad\ $d\leftarrow nv$; $D\leftarrow nm$; $\phi\leftarrow 0$; $u\leftarrow nv$;
    $v\leftarrow nv$; $b\leftarrow nv$; $B\leftarrow nm$; $A\leftarrow I$; 
    
    $\mathtt{4}$\quad\ $\text{NF}\leftarrow \HH_{\alpha}\PP_dZ_D\HH_{\omega}\ee^{\ii\phi}\XX_{u}\ZZ_{v}\PP_bZ_BX_A$;

    $\mathtt{5}\quad\ \mathtt{for}\ k=\ell\ \mathtt{to}\ 1\ \mathtt{do}$
  
    $\mathtt{6}$\quad\quad\ \texttt{if}\  $M_k=\HH_{i}$\ \texttt{then}
    
    $\mathtt{7}$\quad\quad\quad\ \texttt{merge-Hadamard}$(\text{NF},i)$;
    
    $\mathtt{8}\quad\quad\ \mathtt{else\ if\ } M_k=\PP_i\ \mathtt{then}$
    
    $\mathtt{9}$\quad\quad\quad\ \texttt{merge-phase}$(\text{NF},i)$;
    
    $\mathtt{10}\quad\quad \mathtt{else\ }\ $\quad  

    $\mathtt{11}$\quad\quad\quad \texttt{merge-CNOT}$(\text{NF},i,j)$;

    $\mathtt{12}$\quad \texttt{simplifiy}$(\text{NF})$;
    
    $\mathtt{13}$\quad \texttt{return} NF;
    
    \caption{ Algorithm  $\ctonf$ \label{C-to-NF}}
  \end{figure}

  Note that the form \ref{nf} is given as a matrix product and corresponds actually to a quantum circuit in the form CX-CZ-P-Z-X-H-CZ-P-H
  \footnote{The reader who is not used to quantum circuits must pay attention to the following fact: the circuits act to the right of the state $\ket{\psi}$
  presented to their left but the associated operators act to the left of $\ket{\psi}$, so the order of the gates in the circuit is inverted comparing to the form \eqref{nf}.}. To obtain a quantum circuit corresponding to the normal form, one can use the algorithm proposed in 2004 by Patel \textit{et al.} \cite{2004PMH} in order to write the matrix $A$ as a product of transvections.
Our normal form is similar to the forms proposed by Duncan \textit{et al.} \cite{2020DKPV} (H-Z-P-CZ-CX-H-CZ-P-Z-H) or
by Bravyi and Maslov \cite{2020BM} (X-Z-P-CX-CZ-H-CZ-H-P). In these 3 normal forms, the 2-qubit gate count is asymptotically dominated by the two $\cz$ gate layers which contain together up to $n(n-1)$ gates. Indeed, the $\cnot$ gate layer can be decomposed in $O(n^2/\log n)$ $\cnot$ gates thanks to the algorithm by Patel \textit{et al.} \cite{2004PMH}. Therefore it would be interesting to compute another normal form where the gate count is dominated by the $\cnot$ gate layers. To this end we need the following definition and lemma.

\begin{figure}[h]
  $\mathtt{ALGORITHM\ :}$   Reduction of a matrix in $\BG$.
  
  $\mathtt{INPUT\ :}$  $B$, a matrix in $\BG$.
  
  $\mathtt{OUTPUT\ :}$ $(B',A)$, where

    $\quad B'\in\BG$ is a reduced matrix congruent to $B$,

    $\quad A\in\GL$ satisfies the congruence relation $B'=A^{\T}BA$.

    $\mathtt{1}\quad\ B'\leftarrow B;$

    $\mathtt{2}\quad\ A\leftarrow \text{Identity};$

    $\mathtt{3}\quad\ $/* $\mathrm{pivot}[j]=\mathtt{true}$, if $j$ has already been chosen as a pivot */

    $\mathtt{4}\quad\ \mathtt{for}\ j= 0 \ \mathtt{to}\ n-1\ \mathtt{do}$ 

    $\mathtt{5}\quad\quad\ \mathrm{pivot}[j]\leftarrow \mathtt{false}$;

    $\mathtt{6}\quad\ \mathtt{for}\ j= 0 \ \mathtt{to}\ n-2\ \mathtt{do}$ //nothing to do on column $n-1$
  
    $\mathtt{7}\quad\quad\ \mathtt{if}\ \mathrm{pivot}[j]\ \mathtt{or}\ \mathrm{card}\{i\mid B'_{ij}=1\}=0\ \mathtt{then}$
    
    $\mathtt{8}\quad\quad\quad\ \mathtt{continue};$

    $\mathtt{9}\quad\quad\ $/* choosing pivot */
    
    $\mathtt{10}\quad\quad p\leftarrow \mathrm{min}\{ i \mid B'_{ij}=1\};$

    $\mathtt{11}\quad\quad \mathrm{pivot}[p]\leftarrow \mathtt{true};$

    $\mathtt{12}\quad\quad $/* Step a : eliminating the remaining 1's on column $j$ and line $j$ */

    $\mathtt{13}\quad\quad \mathtt{for}\ r=p + 1\ \mathtt{to}\ n-1\ \mathtt{do} $
    
    $\mathtt{14}\quad\quad\quad  \mathtt{if}\ B'_{rj}=1\ \mathtt{then} $

    $\mathtt{15}\quad\quad\quad\quad B'\leftarrow [rp]B'[pr];$

    $\mathtt{16}\quad\quad\quad\quad A\leftarrow A[pr];$

    $\mathtt{17}\quad\quad$/* Step b : eliminating the remaining 1's on line $p$ and column $p$ */
    
    $\mathtt{18}\quad\quad\mathtt{for}\ c=j + 1\ \mathtt{to}\ n-1\ \mathtt{do} $
    
    $\mathtt{19}\quad\quad\quad  \mathtt{if}\ B'_{pc}=1\ \mathtt{then} $

    $\mathtt{20}\quad\quad\quad\quad B'\leftarrow [cj]B'[jc];$

    $\mathtt{21}\quad\quad\quad\quad A\leftarrow A[jc];$

    $\mathtt{22}\quad \mathtt{return} (B',A);$
    
    \caption{ Algorithm : reduction of a matrix in $\BG$\label{B-red-algo}}
  \end{figure}

\begin{defi}
We say that a matrix $B\in\BG$ is \textit{reduced} when each column and each line of $B$ contains at most one non-zero entry, \textit{i.e.} $Z_{B}$ corresponds to a $\cz$ circuit of depth 1.
  \end{defi}

\begin{lem}\label{B-red}
  For any $B$ $\in\BG$, there exists an upper triangular matrix $A\in\GL$ and a reduced matrix $B_{\text{red}}\in\BG$ such that $B_{\text{red}}=A^{\T}BA$.
\end{lem}

\begin{proof}
  $B$ is the matrix of an alternating bilinear form with respect to the canonical basis $(e_i)_{i=0\dots n-1}$.  The equality $B_{\text{red}}=A^{\T}BA$ is just the classical change of basis formula, where $A$ is the matrix of the new basis. A construction of $A$ and $B_{\text{red}}$ is given by the algorithm in Figure \ref{B-red-algo}. We use Gaussian elimination (\textit{i.e.}
  multiplication by transvection matrices, \textit{cf.} Proposition \ref{tij-mult}) on columns and rows of the matrix $B$ to construct step by step the matrices $A$ and $B_{\text{red}}$ (see Example \ref{B-red-example}).
\end{proof}

\begin{example} Computing a reduced matrix in $\BG[7]$. \label{B-red-example}\medskip

  $\mathtt{INPUT} : B=\begin{bmatrix}
    0&0&0&1&0&1&0\\
    0&0&1&1&0&0&1\\
    0&1&0&0&1&1&0\\
    1&1&0&0&1&0&0\\
    0&0&1&1&0&0&0\\
    1&0&1&0&0&0&1\\
    0&1&0&0&0&1&0\\
  \end{bmatrix}$
\medskip

  $\bf j=0 $

  Choosing pivot : $p\leftarrow 3$; $\mathrm{pivot}[3]\leftarrow \mathtt{true};$

  Step a :
$[53]B[35]=
\begin{bmatrix}
  \bf{0}&\bf0&\bf0&\bf1&\bf0&\bf0&\bf0\\
  \bf0&0&1&1&0&1&1\\
  \bf0&1&0&0&1&1&0\\
  \bf1&1&0&0&1&0&0\\
  \bf0&0&1&1&0&1&0\\
  \bf0&1&1&0&1&0&1\\
  \bf0&1&0&0&0&1&0\\
\end{bmatrix}$

Step b : 
$[40][10][53]B[35][01][04]=
\begin{bmatrix}
  \bf0&\bf0&\bf0&\bf1&\bf0&\bf0&\bf0\\
  \bf0&0&1&\bf0&0&1&1\\
  \bf0&1&0&\bf0&1&1&0\\
  \bf1&\bf0&\bf0&\bf0&\bf0&\bf0&\bf0\\
  \bf0&0&1&\bf0&0&1&0\\
  \bf0&1&1&\bf0&1&0&1\\
  \bf0&1&0&\bf0&0&1&0\\
\end{bmatrix}$
\medskip

$\bf j=1$

Choosing pivot : $p\leftarrow 2$; $\mathrm{pivot}[2]\leftarrow \mathtt{true}$;

Step a :
$[62][52][40][10][53]B[35][01][04][25][26]=
\begin{bmatrix}
 \bf0&\bf0&\bf0&\bf1&\bf0&\bf0&\bf0\\
  \bf0&\bf0&\bf1&\bf0&\bf0&\bf0&\bf0\\
  \bf0&\bf1&0&\bf0&1&1&0\\
  \bf1&\bf0&\bf0&\bf0&\bf0&\bf0&\bf0\\
  \bf0&\bf0&1&\bf0&0&0&1\\
  \bf0&\bf0&1&\bf0&0&0&0\\
  \bf0&\bf0&0&\bf0&1&0&0\\
\end{bmatrix}$

Step b : 
$[51][41][62][52][40][10][53]B[35][01][04][25][26][14][15]=
\begin{bmatrix}
 \bf0&\bf0&\bf0&\bf1&\bf0&\bf0&\bf0\\
  \bf0&\bf0&\bf1&\bf0&\bf0&\bf0&\bf0\\
  \bf0&\bf1&\bf0&\bf0&\bf0&\bf0&\bf0\\
  \bf1&\bf0&\bf0&\bf0&\bf0&\bf0&\bf0\\
  \bf0&\bf0&\bf0&\bf0&0&0&1\\
  \bf0&\bf0&\bf0&\bf0&0&0&0\\
  \bf0&\bf0&\bf0&\bf0&1&0&0\\
\end{bmatrix}$
\medskip

$\bf j=2$ : $\mathrm{pivot}[2]=\mathtt{true}$, so $\mathtt{continue}$ 
\medskip

$\bf j=3$ : $\mathrm{pivot}[3]=\mathtt{true}$, so  $\mathtt{continue}$ 
\medskip

$\bf j=4$

Choosing pivot :  $p\leftarrow 6$; $\mathrm{pivot}[6]\leftarrow \mathtt{true}$;

Step a : $B'$ remains unchanged

Step b : $B'$ remains unchanged
\medskip

$\bf j=5$ :  null column, so $\mathtt{continue}$
\medskip

$\mathtt{return} (B',A)$, where
\medskip

$B'=
\begin{bmatrix}
 \bf0&\bf0&\bf0&\bf1&\bf0&\bf0&\bf0\\
  \bf0&\bf0&\bf1&\bf0&\bf0&\bf0&\bf0\\
  \bf0&\bf1&\bf0&\bf0&\bf0&\bf0&\bf0\\
  \bf1&\bf0&\bf0&\bf0&\bf0&\bf0&\bf0\\
  \bf0&\bf0&\bf0&\bf0&\bf0&\bf0&\bf1\\
  \bf0&\bf0&\bf0&\bf0&\bf0&\bf0&\bf0\\
  \bf0&\bf0&\bf0&\bf0&\bf1&\bf0&\bf0\\
\end{bmatrix}$,
$A=
[35][01][04][25][26][14][15]=
\begin{bmatrix}
  1&1&0&0&0&1&0\\
  0&1&0&0&1&1&0\\
  0&0&1&0&0&1&1\\
  0&0&0&1&0&1&0\\
  0&0&0&0&1&0&0\\
  0&0&0&0&0&1&0\\
  0&0&0&0&0&0&1\\
\end{bmatrix}$.

\end{example}

\begin{theo} \textbf{$\cz$-reduced normal form}
  
  Any  stabilizer circuit  $C=\prod_{k=1}^{\ell} M_k $, where $M_k\in\{\PP_i,\HH_i,X_{[ij]}\}$ can be written in the form
  \begin{equation}
    \HH_{\alpha}\PP_dX_{A_{1}}Z_{D_\text{red}}\HH_{\omega}\ee^{\ii\phi}\XX_{u}\ZZ_{v} X_{A_{2}}Z_{B_{\text{red}}}X_{A_{3}}\PP_b\ ,\label{red-nf}
  \end{equation}
  where     $\alpha,\omega,d,b,u,v\in \F^n$, $D_{\text{red}}$ and $B_{\text{red}}$ are reduced matrices in $\BG$,  $A_1,A_2,A_3\in\GL$, $A_1$ and $A_3$ are upper triangular matrices and $\phi\in\{k\frac{\pi}{4},k\in\Z\}$.
   \end{theo}

\begin{proof} The proof consists in the description of an algorithm that takes a stabilizer circuit C as input and return 
  an equivalent circuit in the $\cz$-reduced normal form $\HH_{\alpha}\PP_dX_{A_{1}}Z_{D_\text{red}}\HH_{\omega}\ee^{\ii\phi}\XX_{u}\ZZ_{v} X_{A_{2}}Z_{B_{\text{red}}}X_{A_{3}}\PP_b$.
  This algorithm is called the $\ctorednf$ algorithm. We start by applying the $\ctonf$ algorithm without the final simplification (subroutine \texttt{simplify}) to the input and we obtain a circuit $C=\HH_a\PP_dZ_D\hh\ee^{\ii\phi}\XX_{u}\ZZ_{v}\PP_bZ_BX_A$.
  
  Let $\ZZ_{w'}\PP_{b'}Z_{B'}=X_{A}^{-1}\cdot\PP_bZ_B\cdot X_A$, then
  
  $C=\HH_a\PP_dZ_D\hh\ee^{\ii\phi}\XX_{u}\ZZ_{v}X_A\ZZ_{w'}\PP_{b'}Z_{B'}\stackrel{\ref{conj-xz-XA}}{=}
  \HH_a\PP_dZ_D\hh\ee^{\ii\phi}\XX_{u}\ZZ_{v\oplus A^{-\T}w'}X_AZ_{B'}\PP_{b'}$.

  Let $D_\text{red}$ (resp. $B_{\text{red}})$ be a reduction of $D$ (resp. $B'$). Let $D_\text{red}=A'^{\T} D A'$ and
$B_{\text{red}}=A''^{\T} B' A''$, where $A',A''\in\GL$. Since \ref{conj-ZB-XA} one has $X_{A'}Z_{D_\text{red}}X_{A'}^{-1}=\ZZ_{q_{D_{\text{red}}}(A'^{-1})}Z_{D}$ and
$X_{A''}Z_{B_{\text{red}}}X_{A''}^{-1}=\ZZ_{q_{B_{\text{red}}}(A''^{-1})}Z_{B'}$, hence

$C=\HH_a\PP_d  X_{A'}Z_{D_\text{red}}X_{A'}^{-1}\ZZ_{q_{D_{\text{red}}}(A'^{-1})}   \hh\ee^{\ii\phi}\XX_{u}\ZZ_{v\oplus A^{-\T}w'}
X_A \ZZ_{q_{B_{\text{red}}}(A''^{-1})}X_{A''}Z_{B_{\text{red}}}X_{A''}^{-1}\PP_{b'}$.

After merging the $\ZZ$ gates with $\XX_{u}\ZZ_{v\oplus A^{-\T}w'}$, we obtain :

$C\stackrel{\eqref{conj-z-h},\eqref{conj-xz-XA}}{=}\HH_a\PP_d  X_{A'}Z_{D_\text{red}}X_{A'}^{-1}\hh\ee^{\ii\phi}\XX_{u'}\ZZ_{v'}X_A X_{A''}Z_{B_{\text{red}}}X_{A''}^{-1}\PP_{b'}$,

where $u'=u\oplus q_{D_{\text{red}}}(A''^{-1})$ and $v'=v\oplus A^{-\T}w'\oplus A^{-T}q_{B_{\text{red}}}(A''^{-1})$. Hence

$C=\HH_a\PP_d  X_{A'}Z_{D_\text{red}} \hh\cdot X_{A'^{\T}}\ee^{\ii\phi}\XX_{u'}\ZZ_{v'}X_{A'^{-\T}}\cdot X_{A'^{\T}}X_A X_{A''}Z_{B_{\text{red}}}X_{A''^{-1}}\PP_{b'}$,

$C\stackrel{\ref{conj-xz-XA}}{=}\HH_a\PP_d  X_{A'}Z_{D_\text{red}} \hh\ee^{\ii\phi}\XX_{A'^{\T}u'}\ZZ_{A'^{-1}v'}X_{A'^{\T}AA''}Z_{B_{\text{red}}}X_{A''^{-1}}\PP_{b'}$.

Finally, we build $\HH_{\alpha}$ and $\HH_{\omega}$ after simplifying  $\HH_a$ with $\hh$ as follows. Let $C_i$ (resp. $L_i$) be the column $i$ (resp. the line $i$) of $A'$ and $(e_i)_{i=0\dots n-1}$ be the canonical basis of $\F^n$. Let $Q$ be the set of qubits involved in the subcircuit $\PP_d  X_{A'}Z_{D_\text{red}}$, \textit{i.e}
$Q=\{i\mid d_i=1\}\cup\{i\mid \exists j, D_{\text{red\, }ij}=1\}\cup\{i\mid L_i \neq e_i^{T} \text{ or } C_i \neq e_i\}$. If $a_i=1$ and $i\neq Q$ then $\alpha_i=\omega_i=0$, otherwise $\alpha_i=a_i$ and $\omega_i=1$. 
\end{proof}

Neglecting the global phase $\phi$, the $\cz$-reduced normal form \ref{red-nf} corresponds to a quantum circuit in the form
P-CX-CZ-CX-Z-X-H-CZ-CX-P-H, where each $\cz$ gate layer contains less than
$\frac{n}{2}$ $\cz$ gates. Using the algorithm by Patel \textit{et al.} \cite{2004PMH} to decompose the three $\GL$-matrices $A_i$ in transvections yields a circuit that contains
$O\left(\frac{n^2}{\log(n)}\right)$ two-qubit gates, which is better, in theory, than the $O(n^2)$ two-qubit gate count of the first normal form. Whether or not this second normal form brings a real practical advantage compared to the first one depends, however, on the value of the constant $\mathrm{c}$ such that the two-qubit gate count remains lower than $\mathrm{c}\frac{n^2}{\log(n)}$. In Section \ref{experience} we propose a value for $\mathrm{c}$.
We summarize the algorithms presented in this paper in Table \ref{algos}.

\begin{table}[h]
  \begin{center}
$\begin{array}{|lll|}\hline
   \text{Input circuit of\dots}&\text{Algorithm}&\text{Normal Form}\\
   \hline
  \swap,\cz&\xrightarrow{\ctozs}&Z_BS_{\sigma}\\
  \swap,\cnot,\cz,\PP,\ZZ&\xrightarrow{\ctopzx}&\ZZ_v\PP_bZ_BX_A\\
   \swap,\cnot,\cz,\PP,\ZZ,\XX,\YY,\HH&\xrightarrow{\ctonf}&\HH_{\alpha}\PP_dZ_D\HH_{\omega}\ee^{\ii\phi}\XX_{u}\ZZ_{v}\PP_bZ_BX_A\\
   \swap,\cnot,\cz,\PP,\ZZ,\XX,\YY,\HH&\xrightarrow{\ctorednf}&\HH_{\alpha}\PP_dX_{A_{1}}Z_{D_\text{red}}\HH_{\omega}\ee^{\ii\phi}\XX_{u}\ZZ_{v}
                                                                                    X_{A_{2}}Z_{B_{\text{red}}}X_{A_{3}}\PP_b\\
   \hline
 \end{array}$
 \end{center}
 \caption{Algorithms and normal forms\label{algos}}
 \end{table}

\section{Application to stabilizer states and graph states\label{application}}

A stabilizer state $\ket{S}$ for a $n$-qubit register can be written in the form $\ket{S}=C\ket{0}^{\otimes n}$ where $C$ is a stabilizer circuit
\cite[Theorem 1]{2004AG}. A graph state $\ket{G}$ is a special case of a stabilizer state that can be written in the form
$\ket{G}=Z_B\ket{+}^{\otimes n}=Z_B\hh\ket{0}^{\otimes n}$,
where $\ket{+}=\HH\ket{0}=\frac{1}{\sqrt{2}}(\ket{0}+\ket{1})$ is the eigenvector corresponding to the eigenvalue 1 of the \PauliX gate and $\hh=\HH^{\otimes n}$\cite{2006HD}.
The graph $G$ associated to the graph state $\ket{G}$ is the graph of order $n$ whose vertices are labeled by the $n$ qubits and whose set of edges is $B=\{\{i,j\}\mid B_{ij}=1\}$.
Let $C=\prod_{k=1}^{\ell} M_k$ be the product of $\ell$ Clifford gates $M_k\in\{\PP_i,\HH_i,X_{[ij]}\}$. Applying the $\ctonf$ algorithm to $C$ without performing the final simplification (subroutine \texttt{simplify}) yields $C=\HH_a\PP_dZ_D\hh\ee^{\ii\phi}\XX_{u}\ZZ_{v}\PP_bZ_BX_A\stackrel{\eqref{conj-z-h}}{=}
\ee^{\ii\phi}\HH_a\ZZ_{u}\PP_dZ_D\hh\ZZ_{v}\PP_bZ_BX_A$. Since the subcircuit $\ZZ_{v}\PP_bZ_BX_A$  has no effect on the ket $\ket{0}^{\otimes n}$, one has (neglecting the global phase $\phi$)  
$\ket{S}=\HH_a\ZZ_{u}\PP_dZ_D\hh\ket{0}^{\otimes n}=\HH_a\ZZ_{u}\PP_d\ket{G}$, where $\ket{G}$ is the graph state $Z_D\hh\ket{0}^{\otimes n}$. So, using the $\ctonf$ algorithm, we obtain a proof of the well known statement that any stabilizer state $\ket{S}$ is equivalent to a graph state $\ket{G}$ under local Clifford operations : there exists a stabilizer circuit $C'$ composed of local Clifford gates only (\textit{i.e.} phase and Hadamard gates) such that  $\ket{S}=C'\ket{G}$  (see \cite[theorem 1]{2004VDN}). Moreover our method provides straightforwardly  a possible value for the circuit $C'$ and the graph $G$.
\begin{prop}\label{stab-graph}
  For any stabilizer state $\ket{S}$, there exists a graph state $\ket{G}$ and 3 vectors $u,v,w$ in $\F^n$ such that
\begin{equation}
  \ket{S}=\HH_u\ZZ_v\PP_w\ket{G}.
  \end{equation}
\end{prop}

 Applying the $\ctorednf$ algorithm to the special case of a  $\cz$-gate circuit $Z_B$ ($B\in\BG$) yields
 \begin{equation}
   Z_B=\ZZ_vX_AZ_{B_{\text{red}}}X_{A^{-1}},\label{ZBred}
 \end{equation}
 where $A$ is an upper triangular matrix in $\GL$ and $B_{\text{red}}$ is a reduced matrix in $\BG$. 
Using \eqref{ZBred}, we write the graph state $\ket{G}=Z_B\hh\ket{0}^{\otimes n}$ in the form
$\ket{G}=\ZZ_vX_AZ_{B_{\text{red}}}X_{A^{-1}}\hh\ket{0}^{\otimes n}$. Since \eqref{conj-XA-h}, we obtain
$\ket{G}=\ZZ_vX_AZ_{B_{\text{red}}}\hh X_{A^{\T}}\ket{0}^{\otimes n}$. As a $\cnot$ circuit has no effect on the ket $\ket{0}^{\otimes n}$,
one has $\ket{G}=\ZZ_vX_AZ_{B_{\text{red}}}\hh\ket{0}^{\otimes n}$. Hence the proposition that follows.
\begin{prop}\label{graph-state}
  Any graph state $\ket{G}$ can be written in the form
\begin{equation}
  \ket{G}=\ZZ_vX_AZ_{B_{\text{red}}}\ket{+}^{\otimes n},
\end{equation}
where $v\in\F^n$, $A\in\GL$ is an upper triangular matrix, $Z_{B_{\text{red}}}$ is a $\cz$ circuit of depth 1. 
\end{prop}



    \begin{example}
      The entangled state of a 5-qubit register $\ket{\mathtt{GHZ}}_5$ can be easily implemented as a stabilizer state :
      $\ket{\mathtt{GHZ}}_5=X_{[43]}X_{[32]}X_{[21]}X_{[10]}\HH_0\ket{0}^{\otimes n}$. Applying the $\ctonf$ algorithm on the input $C=X_{[43]}X_{[32]}X_{[21]}X_{[10]}\HH_0$  yields the following normal form :
      $C=\HH_1\HH_2\HH_3\HH_4Z_{01}Z_{02}Z_{03}Z_{04}\HH_0\HH_1\HH_2\HH_3\HH_4X_{[21]}X_{[31]}X_{[41]}X_{[43]}X_{[32]}$ (the reader can use our C implementation of the algorithm, see next section).
      So $\ket{\mathtt{GHZ}}_5=\HH_1\HH_2\HH_3\HH_4Z_{01}Z_{02}Z_{03}Z_{04}\hh\ket{0}^{\otimes n}$, hence $\ket{\mathtt{GHZ}}_5$ is Local Clifford equivalent to the star graph state 
      $\ket{G}=Z_{\{\{0,1\},\{0,2\},\{0,3\},\{0,4\}\}}\ket{+}^{\otimes n}$ (Proposition \ref{stab-graph}). Then we use the $\ctorednf$ algorithm on the input $Z_{\{\{0,1\},\{0,2\},\{0,3\},\{0,4\}\}}$ and we obtain
      
      $Z_{\{\{0,1\},\{0,2\},\{0,3\},\{0,4\}\}}=X_{[14]}X_{[13]}X_{[12]}Z_{01}X_{[14]}X_{[13]}X_{[12]}$, hence $\ket{G}=X_{[14]}X_{[13]}X_{[12]}Z_{01}\ket{+}^{\otimes n}$ (Proposition \ref{graph-state}).
      \end{example}

 \section{Empirical validation\label{experience}}
 We implement the algorithms presented in this paper in the C language with a text-based user interface. The source code is available at :
 
 \url{https://github.com/marcbataille/stabilizer-circuits-normal-forms}.
 
 To decompose the $\GL$-matrices of the normal forms, we use  the algorithm by Patel \textit{et al.} with a value of $\lceil\log_2(n)/2\rceil$ for the parameter $m$ (see  \cite{2004PMH}). Our program is fast and can write in normal form a 40000 gates random stabilizer circuit for a 200-qubit register in a few seconds using a basic laptop. The \textit{manual} mode of the program reproduces the induction steps of the $\ctonf$ algorithm, while the \textit{statistics} mode provides convenient tools to show empirically the interest (and the limits) of our normal forms to reduce stabilizer circuits.  A few significant results are presented below and we invite the reader to use our program to obtain his own statistics. We use samples of $100$ random stabilizer circuits and we compute the average length and the average 2-qubit gate count of both normal forms (results are given in percentage of the input).
 The probability of choosing a $\cnot$ gate is $0.8$ while the probability of choosing a phase or a Hadamard gate is 0.1 for each type (these proportions can be easily modified in the program).
 \medskip
 \begin{center}

 \begin{tabular}{|c|c|c|c||c|c|c|}
   \hline
   \multicolumn{7}{|c|}{Input circuits of length $n^2/2$}\\ \hline  \hline
   &\multicolumn{3}{c||}{All gate count}&\multicolumn{3}{c|}{2-qubit gate count}\\ \hline
   n &Input& N.F.&$\cz$-red. N.F.&Input& N.F.&$\cz$-red. N.F. \\    \hline
   10&100\%&165\%&172\%&100\%&153\%&157\%  \\    \hline
   20&100\%&199\%&188\%&100\%&208\%&193\%  \\    \hline
50&100\%&181\%&152\%&100\%&209\%&173\%  \\    \hline
   100&100\%&166\%&129\%&100\%&200\%&153\%  \\    \hline
   200&100\%&157\%&113\%&100\%&193\%&137\%  \\    \hline
   300&100\%&152\%&103\%&100\%&188\%&126\%  \\    \hline
   \end{tabular}

   \medskip

 \begin{tabular}{|c|c|c|c||c|c|c|}
   \hline
   \multicolumn{7}{|c|}{Input circuits of length $n^2$}\\ \hline  \hline
   &\multicolumn{3}{c||}{All gate count}&\multicolumn{3}{c|}{2-qubit gate count}\\ \hline
   n &Input& N.F.&$\cz$-red. N.F.&Input& N.F.&$\cz$-red. N.F. \\    \hline
   10&100\%&114\%&113\%&100\%&104\%&101\%  \\    \hline
   20&100\%&104\%&96\%&100\%&109\%&98\%  \\    \hline
   50&100\%&90\%&76\%&100\%&105\%&87\%  \\    \hline
   100&100\%&83\%&64\%&100\%&100\%&76\%  \\    \hline
   200&100\%&79\%&56\%&100\%&96\%&68\%  \\    \hline
   300&100\%&76\%&52\%&100\%&94\%&63\%  \\    \hline
  \end{tabular}

  \medskip

 \begin{tabular}{|c|c|c|c||c|c|c|}
   \hline
   \multicolumn{7}{|c|}{Input circuits of length $2n^2$}\\ \hline  \hline
    &\multicolumn{3}{c||}{All gate count}&\multicolumn{3}{c|}{2-qubit gate count}\\ \hline
   n &Input& N.F.&$\cz$-red. N.F.&Input& N.F.&$\cz$-red. N.F. \\    \hline
   10&100\%&61\%&60\%&100\%&55\%&54\%  \\    \hline
   20&100\%&52\%&48\%&100\%&54\%&49\%  \\    \hline
   50&100\%&45\%&38\%&100\%&52\%&43\%  \\    \hline
   100&100\%&42\%&32\%&100\%&50\%&38\%  \\    \hline
   200&100\%&39\%&28\%&100\%&48\%&34\%  \\    \hline
   300&100\%&38\%&26\%&100\%&47\%&32\%  \\    \hline
 \end{tabular}
\end{center}
 \medskip

 The experimental results clearly show that the $\ctorednf$ algorithm is better than the $\ctonf$ algorithm in terms of circuit reduction.
 From $20$ qubits,  we observe that the $\ctorednf$ algorithm is able to reduce the 2-qubit gate count, not only in average, but for all circuits of length $n^2$ tested in the experiment (second table).
 In fact, this observation is just a particular case of the conjecture that follows.
\begin{conj}
 Using the $\ctorednf$ algorithm, any $n$-qubit stabilizer circuit can be transformed  into an equivalent circuit that contains less than
 $\dfrac{3n^2}{\log(n)}$ two-qubit gates. 
 \end{conj}
 
 We check this conjecture up to 600 qubits on about one thousand random circuits of different lengths.\medskip

We focus now on the particular case of $\cz$ circuits and we evaluate to what extent Identity \ref{ZBred} is helpful to reduce the 2-qubit gate count in an implementation of a $\cz$ circuit (corresponding to a matrix $Z_B$) that is based on  \PauliZ, $\cnot$ and $\cz$ gates.
Again we use the algorithm by Patel \textit{et al.} with a value of $\lceil\log_2(n)/2\rceil$ for the parameter $m$ in order to decompose the matrix $A$ and we obtain an equivalent circuit to $Z_B$. It appears that such an implementation of a $\cz$ circuit is interesting in some cases from about 20 qubits. For instance a 20-qubit $\cz$ circuit corresponding to a matrix $B$ with all non-diagonal entries equal to 1 (complete graph $K_{20}$, 190 $\cz$ gates) is reduced to an equivalent circuit of 164 two-qubit gates
($\cnot$ + $\cz$). For a high number of qubits, the reduction is more interesting. The table below shows results for $200$-qubit $\cz$ circuits. We use samples of 100 random circuits of the same length. All input circuits contain distinct $\cz$ gates and we count only the 2-qubit gates ($\cnot$ + $\cz$) of the output circuit.
\medskip

\begin{tabular}{|c|c|c|c|c|c|c|c|c|}
  \hline
  \multicolumn{9}{|c|}{200-qubit $\cz$ circuits}\\ \hline  \hline
  \multicolumn{2}{|c|}{Input length} & 5000&10000&12000&14000&16000&18000&19000\\ \hline \hline
  \multirow{3}{*}{Output}& Max. &216.2 \%&108.9 \%&91.0 \%&77.9 \%&68.1 \%&60.1 \%&55.7 \%\\ 
  & Avg.&214.8 \%&108.4 \%&90.4 \%&77.4 \%&67.6 \%&59.4 \%&54.2 \%\\ 
  & Min.&213.5 \%&107.8 \%&89.7 \%&76.9 \%&67.2 \%&58.2 \%&52.7 \%\\ \hline
\end{tabular}\medskip

We end this section by a simple example which highlights the usefulness that can have  Proposition \ref{graph-state}, and more generally the $\ctorednf$ algorithm, to reduce the gate count
of stabilizer circuits implemented in a real-life quantum machine. We implement in the publicly available 5-qubit ibmq\_belem device
(\url{https://quantum-computing.ibm.com/}) the graph state
\begin{equation}
  \ket{K_5}=Z_B\ket{+}^{\otimes n},\label{gfull}
\end{equation}
where $B$ is the adjacency matrix of the complete graph $K_5$.

This graph-state is Local Clifford equivalent to $\ket{\mathtt{GHZ}}_5$, just like the star graph state
$Z_{\{\{0,1\},\{0,2\},\{0,3\},\{0,4\}\}}\ket{+}^{\otimes n}$ mentioned in the previous section (see \textit{e.g.} \cite[Section 4.1]{2006HD}). 
 From Proposition \ref{graph-state}, we obtain (using our computer program) :
 \begin{equation}
   \ket{K_5}=\ZZ_2\ZZ_3X_{[34]}X_{[23]}X_{[12]}X_{[02]}X_{[24]}X_{[23]}Z_{01}Z_{23}\ket{+}^{\otimes n}.\label{gfull-red}
 \end{equation}
 Observe that the form \ref{gfull-red} contains only 8 two-qubit gates comparing to the 10 $\cz$ gates of the form \ref{gfull}. The reduction is much more impressive when
 we consider the circuit that is really implemented in the quantum computer using only native gates. Indeed,  the $\cz$ gate
 is not native in the IBM quantum devices and is simulated thanks to Identity \eqref{zijxij}. The Hadamard gate is implemented from the $R_z(\pi/2)$ and $\sqrt{X}$ gates, since
 $\HH=\ee^{\ii\frac{\pi}{4}}R_z(\pi/2)\sqrt{X}R_z(\pi/2)$, where $\sqrt{X}=\dfrac12\begin{bmatrix}1+\ii&1-\ii\\1-\ii&1+\ii\end{bmatrix}$ and
   $R_z(\theta)=\begin{bmatrix}\ee^{-\ii\frac{\theta}{2}}&0\\
     0&\ee^{\ii\frac{\theta}{2}}
   \end{bmatrix}$. Moreover full connectivity is not achieved and the direct connections allowed between two qubits are given by a graph. The graph of the 5-qubit ibmq\_belem device is $\{\{1,0\},\{1,2\},\{1,3\},\{3,4\}\}$.
   So, to implement a $\cnot$ gate between qubits without direct connection (\textit{e.g.} qubits 2 and 3), it is necessary to simulate it from the native $\cnot$ gates using methods we described on a previous work \cite[Section 3]{2020B}. Due to its similarities to the compilation process in classical computing, the rewriting process that transforms an input circuit with measurements into a native gate circuit giving statistically the same measurement results, is called \textit{transpilation} on the IBM quantum computing website. We remark that the transpiled circuit corresponding to the form \ref{gfull} contains 43 $\cnot$ gates and 69 single qubit gates which is far more than the 16 $\cnot$'s and the 17 single gates of the transpiled circuit corresponding to the form \ref{gfull-red} (see circuits below).\bigskip
   
   $\ket{K_5}=Z_B\ket{+}^{\otimes n}$ :
   
   \includegraphics[scale=0.4, viewport=0cm 0cm 30.5cm 11cm, clip=true]{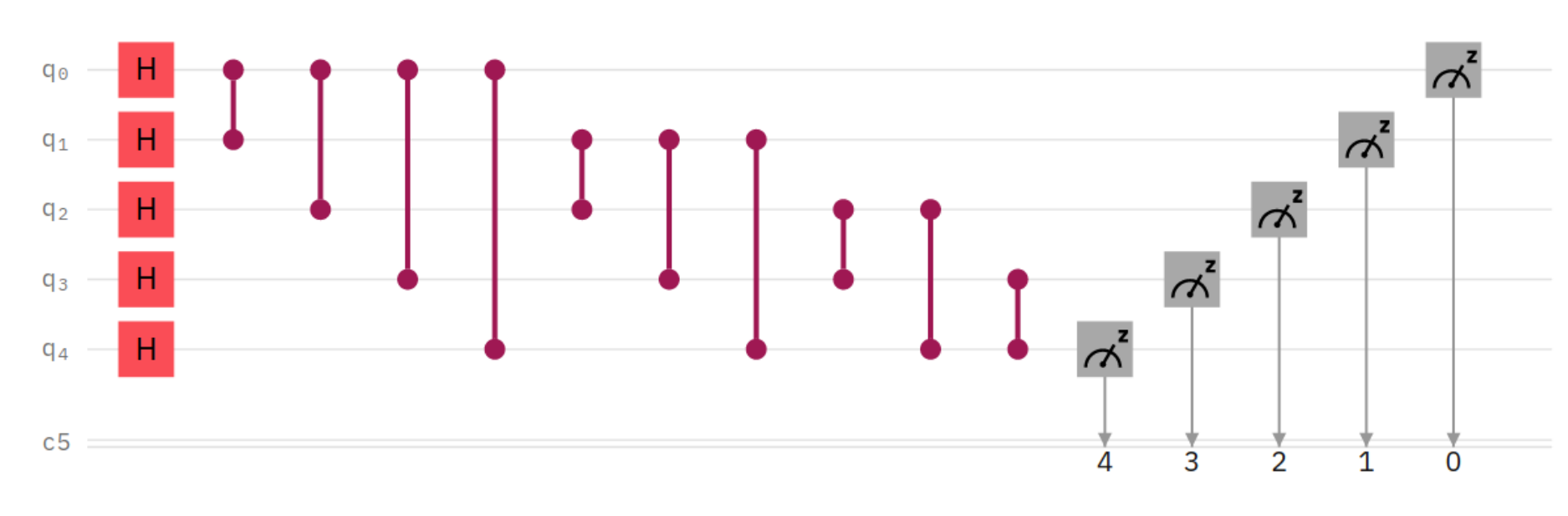}\raisebox{2.1cm}{$\qquad\xrightarrow{transpilation}$}
   
   \includegraphics[scale=0.4,viewport=0cm 0cm 33.9cm 11cm, clip=true]{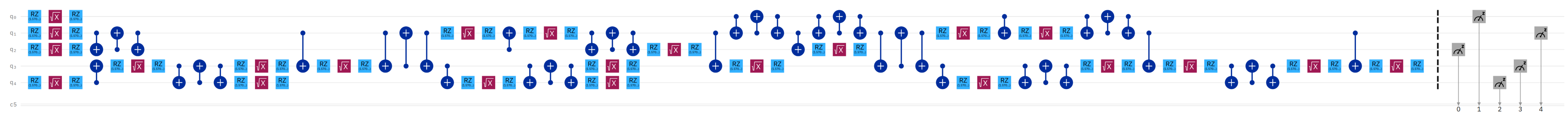}\raisebox{2.45cm}{$\quad\dots$}

   \raisebox{2.45cm}{$\dots\quad$}\includegraphics[scale=0.4,viewport=33.9cm 0cm 67.3cm 11cm, clip=true]{circuit_ghz5_full_trans}\raisebox{2.45cm}{$\quad\dots$}
   
   \raisebox{2.45cm}{$\dots\quad$}\includegraphics[scale=0.4,viewport=67.3cm 0cm 100.9cm 11cm, clip=true]{circuit_ghz5_full_trans}\raisebox{2.45cm}{$\quad\dots$}

   \raisebox{2.45cm}{$\dots\quad$}\includegraphics[scale=0.4,viewport=100.9cm 0cm 133cm 11cm, clip=true]{circuit_ghz5_full_trans}
   
   $\ket{K_5}=\ZZ_2\ZZ_3X_{[34]}X_{[23]}X_{[12]}X_{[02]}X_{[24]}X_{[23]}Z_{01}Z_{23}\ket{+}^{\otimes n}$ :

   \includegraphics[scale=0.4, viewport=0cm 0cm 29cm 11cm, clip=true]{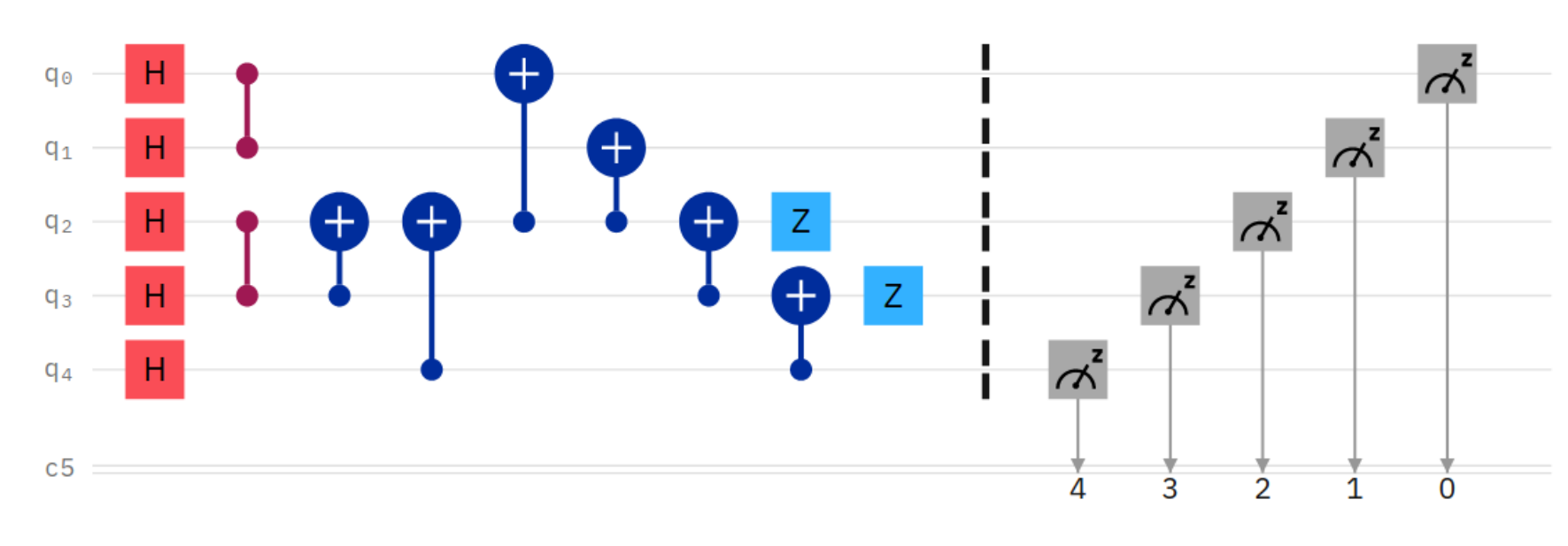}\raisebox{2.1 cm}{$\qquad\xrightarrow{transpilation}$}
   
   \includegraphics[scale=0.4, viewport=0cm 0cm 34cm 11cm, clip=true]{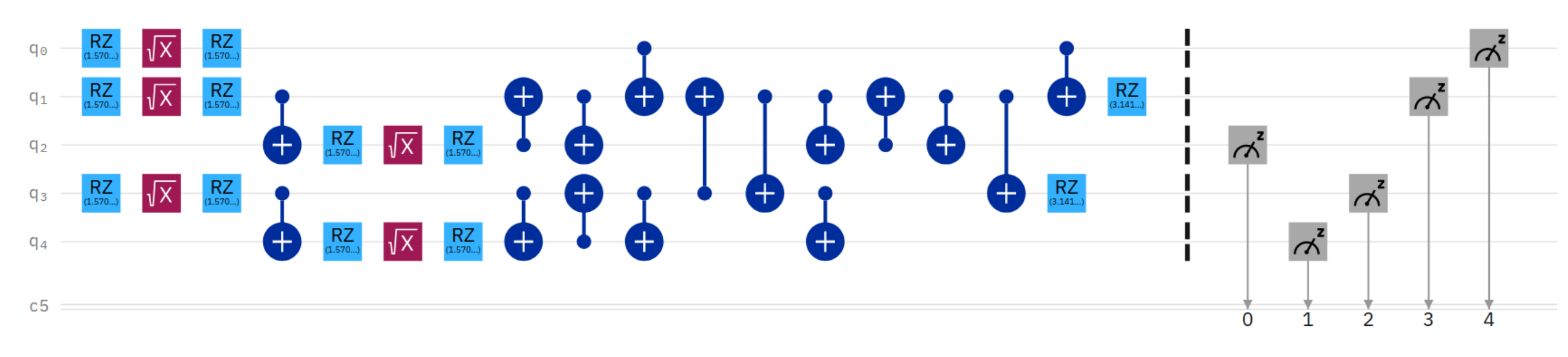}\raisebox{2.45cm}{$\quad\dots$}

   \includegraphics[scale=0.4, viewport=34cm 0cm 50cm 11cm, clip=true]{circuit_ghz5_full_czred_trans}
   
\section{Conclusion and  future work}
Gottesman proved in his Phd thesis that any unitary matrix in the Clifford group is uniquely defined, up to a global phase, by its action by conjugation on the Pauli gates $\XX_i$ and $\ZZ_i$ \cite[pp.41,42]{1997G}. This central statement of Gottesman stabilizer formalism can be used to compute normal forms for $n$-qubit stabilizer circuits via the symplectic group over $\F$ in dimension $2n$ (\textit{e.g.} \cite{2004AG,2018MR}).
In this paper we showed that it is possible to compute normal forms in polynomial time without using this formalism. The reader who is used to work with the symplectic group will notice that our induction process can also be applied inside this group, giving rise to a decomposition of type $M_{\sigma}\begin{bmatrix}I_n&0\\D&I_n\end{bmatrix}  \begin{bmatrix}I_n&B\\0&I_n\end{bmatrix}\begin{bmatrix}A&0\\0&A^{-\T}\end{bmatrix}$ for the symplectic matrix associated to the form \ref{nf},
where $B$ (resp. $D$) is a symmetric matrix corresponding to $\PP_bZ_B$ (resp. $\PP_dZ_D$), $A\in\GL$ is the invertible matrix corresponding to the $\cnot$ subcircuit $X_A$, and $M_{\sigma}$ is a permutation matrix corresponding to a circuit of Hadamard gates.
\medskip

Along with this article we also provided a C implementation of all our algorithms as well as a few basic statistics that help to understand how normal forms can be helpful to reduce the gate count of stabilizer circuits. We applied our results to graph states and we checked experimentally the practical utility of normal forms to implement this type of stabilizer state on real-life quantum computers.\medskip

In \cite{2020B} (resp. \cite{2019BL}), we studied the emergence of entanglement in particular stabilizer circuits, namely $\cnot$ (resp.  $\cz$ plus $\swap$) circuits. Our goal was to find out what kind of entanglement can be created when those simple  circuits act on a fully factorized state. It would be interesting to understand how well normal forms could help to extend these studies to any stabilizer circuit. We leave this topic for future work.

\bibliographystyle{plain}
\bibliography{biblio_MB}

\end{document}